\documentclass[lettersize,onecolumn]{IEEEtran}
\usepackage{amsmath,amsfonts}
\usepackage{algorithmic}
\usepackage{algorithm}
\usepackage{array}
\usepackage[caption=false,font=normalsize,labelfont=sf,textfont=sf]{subfig}
\usepackage{textcomp}
\usepackage{stfloats}
\usepackage{url}
\usepackage{verbatim}
\usepackage{graphicx}
\usepackage{cite}
\usepackage{threeparttable}
\usepackage{color}
\usepackage{amsthm}
\usepackage{amssymb}
\usepackage{blkarray}
\usepackage{booktabs}
\usepackage{caption}
\usepackage{multirow}
\newcommand{\F}{\mathbb{F}}

\newtheorem{theorem}{Theorem}

\newtheorem{definition}{Definition}
\newtheorem{exmpl}{Example}
\newtheorem{lemma}{Lemma}

\newtheorem{remark}{Remark}
\newtheorem{construction}{Construction}

\begin{document}

\title{The Asymmetric Hamming Bidistance and Distributions over Binary Asymmetric Channels}

\author{Shukai Wang, Cuiling Fan, Chunming Tang and Zhengchun Zhou}

\maketitle

\begin{abstract}
The binary asymmetric channel is a  model for practical communication systems where the error probabilities for symbol transitions $0\rightarrow 1$ and $1\rightarrow0$ differ substantially. 
In this paper, we introduce the notion of asymmetric Hamming bidistance (AHB) and its two-dimensional distribution, 
which separately captures directional discrepancies between codewords.
This finer characterization enables a more discriminative analysis of decoding error probabilities for maximum-likelihood decoding (MLD), 
particularly when conventional measures, such as weight distributions and existing discrepancy-based bounds, fail to distinguish code performance. 
Building on this concept, we derive a new upper bound on the average error probability for binary codes under MLD and show that, in general, 
it is incomparable with the two existing bounds derived by Cotardo and Ravagnani (IEEE Trans. Inf. Theory, 68 (5), 2022).
To demonstrate its applicability, we compute the complete AHB distributions for several families of codes, 
including two-weight and some three-weight projective codes (with the zero codeword removed) via strongly regular graphs and 3-class association schemes, 
as well as nonlinear codes constructed from symmetric balanced incomplete block designs (SBIBDs).
\end{abstract}

\begin{IEEEkeywords}
Binary asymmetric channel, asymmetric Hamming bidistance, strongly regular graph, $3$-class association scheme, SBIBD.
\end{IEEEkeywords}

\section{Introduction}\label{sec:1}

\IEEEPARstart{T}{he} study of the {\it binary asymmetric channel}  (BAC), initiated in the 1950s \cite{S1955}, 
centers on a communication model in which the two transmitted symbols $\{0,1\}$ are characterized by distinct crossover probabilities. 
As a foundational yet practically significant construct in information theory and digital communications, the BAC captures inherent directional asymmetries in error susceptibility that are not represented in symmetric channel models. 
Such asymmetries are observed in various practical systems -- including optical links \cite{optical2007}, flash memory devices \cite{CGO1999,CSB2010}, and neuroscience \cite{CIM2013,ZJF2023,SSL2026,ZWC2025} -- where the probability of a $0\rightarrow1$ error often substantially differs from that of a $1\rightarrow0$ error. 
This structural disparity necessitates a departure from classical symmetric analytical frameworks and motivates the development of specialized coding schemes, decoding rules, and performance bounds tailored to asymmetric conditions 
(see \cite{P2016,Q2019,QCR2018,rep2010} and the references therein).

Specifically, the BAC is a discrete memoryless channel defined over $\F_2=\{0,1\}$ with transition probabilities given by
\begin{equation}\label{def-P}
\begin{split}
   \Pr(1|0)=p, &~~\Pr(0|0)=1-p, \\
   \Pr(0|1)=q, &~~\Pr(1|1)=1-q,
\end{split}
\end{equation}
where $\Pr(a|b)$ denotes the probability of receiving $a$ if $b$ was transmitted.
The parameters $p$ and $q$ are drawn from the following set \cite{QCR2018}:
\begin{equation*}
\mathcal{T}=\{(p,q)\in[0,1]^2:p\le q,~p+q<1\}.
\end{equation*}
As reflected in $\mathcal{T}$, the BAC generalizes both the \emph{binary symmetric channel} ($p=q$) and the \emph{Z-channel} ($p=0$).
In this paper, we focus on the regime $0<p\le q<1/2$, which is the standard setting adopted in \cite{CIM2013, GR2022} and related works.

Indeed, numerous studies on asymmetric channels have centered on coding properties \cite{CR1979} and on the design of codes with specific attributes, 
targeting either general asymmetric channels or the Z-channel, as discussed in \cite{GD2012, CR1979, TK1981, MGK1998, XL2005}. 
Despite the merits of the codes examined in these works, the decoding metrics they employ are generally not suitable for {\it maximum likelihood decoding} (MLD) \cite{QCR2018}. 
In response, the authors of \cite{QCR2018} focused on general binary asymmetric channels -- excluding the binary symmetric channel and the Z-channel -- from an MLD perspective.
They further explored the channel equivalence problem via the so-called BAC-function (see \cite[Definition 5]{QCR2018}).

MLD is of fundamental importance in digital communication systems, as it provides the optimal decision rule for minimizing the probability of decoding error when codewords are equally likely to be transmitted. 
The key performance metric of interest for MLD is the {\it average error probability}, which quantifies the overall reliability of the communication link. 
However, the exact computation of the average error probability is generally intractable due to the exponential growth of the decision regions with code length and the dependence on the specific code structure. 
Instead, analytical performance evaluation often relies on the union bound based on {\it pairwise error probabilities}, providing a manageable yet informative upper bound on the error rate.

As early as the 1960s, the decoding error probability of the binary symmetric channel was studied in relation to {\it weight distributions} of codes\cite{book1968}. 
Subsequent research has expanded and refined this line of inquiry, yielding numerous relevant results (see, e.g., \cite{P1994,DCC2006} and references therein). 
For the binary asymmetric channel (BAC), investigations into decoding error probability have also been carried out \cite{rep2010,GR2022,QCR2018}. 
In particular, building on partial results from \cite{QCR2018}, the authors of \cite{GR2022} introduced a channel parameter $\gamma = \log_{(q/(1-p))}\big(p/(1-q)\big)$ 
and defined two notions of discrepancy between binary vectors, denoted by $\delta_\gamma$ and $\hat{\delta}_\gamma$ respectively. 
Using these measures together with classical weight distributions, they derived two distinct upper bounds on the average error probability of maximum-likelihood decoding (see Lemma \ref{lem:bound}).
However, these bounds exhibit certain limitations: when two codes share identical weight distributions and minimum discrepancy values, the bounds fail to distinguish their performance. 
To overcome this drawback, a more refined characterization of binary codes is required.

In this paper, we introduce a new dissimilarity measure for binary codewords ${\bf x}$ and ${\bf y}$, termed the asymmetric Hamming bidistance (see Definition \ref{AHB}), denoted as
\[d_A({\bf x},{\bf y})=(d_{10}({\bf x},{\bf y}),d_{01}({\bf x},{\bf y})),\]
where $d_{ab}({\bf x},{\bf y})=|\{i:x_i=a, y_i=b\}|$. 
This measure is a two-dimensional vector whose components count the number of positions where the symbols differ in each direction. 
In the context of the binary asymmetric channel, the conventional Hamming distance fails to capture the directional asymmetry of error probabilities, 
as it only accounts for the total number of differing symbols. 
In contrast, $d_A$ provides a finer-grained characterization of the channel's asymmetric behavior, which in certain cases allows for a more accurate analysis of decoding error probabilities. 
By leveraging the two-dimensional distance distribution derived from $d_A$, a bound with enhanced discriminative power can be established, thereby improving the performance estimation of codes.
Furthermore, we compare our bound with those obtained in \cite{GR2022} using the discrepancy measures $\delta_\gamma$ and $\hat\delta_\gamma$, 
and clarify the distinct advantages offered by the bidistance approach. Generally, our bound and the known two discrepancy-based bounds are incomparable.

However, completely determining the two-dimensional distribution of $d_A$ for general binary codes is a highly challenging task, 
as the computational complexity grows rapidly with code length and codebook size, often rendering exact analysis infeasible. 
Therefore, another key contribution of this work is the complete characterization of the two-dimensional distribution of the asymmetric Hamming bidistance for several classes of few-weight codes, and some special nonlinear codes constructed from symmetric balanced incomplete block designs (SBIBDs). 
These results can provide a reliable theoretical basis for analyzing the performance of such codes in asymmetric channels.

The remainder of this paper is organized as follows.
Section \ref{sec:2} reviews the necessary preliminaries, including two known bounds on the average decoding error probability from\cite{GR2022}, and a discussion of their limitations. 
In Section \ref{sec:3}, we introduce the asymmetric Hamming bidistance (AHB) and its distribution, derive a new upper bound, and compare it with the existing bounds.
Section \ref{sec:4} is devoted to the computation of AHB distributions for two-weight and some three-weight projective codes (excluding the zero codeword), 
using strongly regular graphs and 3-class association schemes, respectively.
Section \ref{sec:6} extends this analysis to some SBIBD-derived nonlinear codes. Finally, Section \ref{sec:7} concludes the paper.

\section{Preliminaries} \label{sec:2}

In this section, we recall several definitions and preliminary results that will be used throughout the paper, 
including binary codes, the maximum likelihood decoding (MLD) and the two known bounds for the average error probability of codes.
Throughout the paper, $\F_2$ is the binary field, and $n\geq2$ is an integer. 

\subsection{Binary Codes}
A binary code of length $n$ is a subset $C\subseteq\F_2^n$ with $|C| \ge 2$, whose elements are called {\it codewords}. 
For any codeword ${\bf c}=(c_1,c_2,\ldots,c_{n})\in C$, the set of its nonzero positions $\{1\le i \le n:c_i=1\}$ is called the \emph{support set} of ${\bf c}$, denoted as ${\rm supp}(\mathbf{c})$, 
and the size of ${\rm supp}(\mathbf{c})$ is called the {\it weight} of ${\bf c}$, denoted as ${\rm wt}(\mathbf{c})$. 
The {\it Hamming distance} between ${\bf x},{\bf y}\in\F_2^n$ is defined to be $d_H({\bf x},{\bf y})=wt({\bf x}+{\bf y})$. 

Let $A_i$ be the number of codewords in $C$ with weight $i$ for $0\leq i\leq n$.
The \emph{weight enumerator} of $C$ is defined by
$$A_0+A_1z+A_2z^2+\cdots+A_nz^n.$$
The sequence $(A_0,A_1,A_2,\ldots,A_n)$ is the \emph{weight distribution} of $C$.
A code $C$ is said to be a $t$-weight code if the number of nonzero $A_i$ in $(A_1,A_2,\ldots,A_n)$ is equal to $t$.

If $C$ is a $k$-dimensional linear subspace of $\F_2^n$, then $C$ is called a $[n,k,d]$ linear code, with $d=\min\limits_{\forall{\bf x}\neq{\bf y}\in C}d_H({\bf x},{\bf y})$. 
The {\it dual code} of $C$ is defined to be the orthogonal subspace $C^{\bot}$ of $C$ with respect to the Euclidean inner product, i.e., 
$$C^{\bot}=\{\mathbf{c}^{\bot}\in \F_2^n:\mathbf{c}^{\bot}\cdot \mathbf{c}=0 \text{ for all } \mathbf{c}\in C\}.$$
Clearly, the dimension of $C^{\bot}$ is $n-k$. 
A linear code $C$ is said to be \emph{projective} if $d(C^{\bot})\geq 3$.

\subsection{Maximum Likelihood Decoding and Error Probability}

Let $C\subseteq \F_2^n$ be a binary code. 
The conditional probability $\Pr({\bf y}|{\bf x})$, often termed the {\it likelihood function}, 
quantifies the probability of observing the received vector $\mathbf{y}=(y_1,y_2,\ldots,y_n) \in \F_2^n$ given that the codeword $\mathbf{x}=(x_1,x_2,\ldots,x_n)\in C$ was transmitted. 
For a discrete memoryless channel (DMC), this probability factorizes as:
\begin{equation}\label{eq:1}
\Pr(\mathbf{y}|\mathbf{x})=\prod_{i=1}^n \Pr(y_i|x_i),
\end{equation}
where $\Pr(y_i|x_i)$ denotes the channel probability for a single symbol. 
This factorization follows from the memoryless property: the channel output at time $i$ depends only on the input at time $i$, and not on previous or future transmissions.

In the decoding process, a natural way to decode a received message ${\bf y}$ is to return the unique codeword ${\bf x}\in C$ that maximizes $\Pr({\bf y}|{\bf x})$, or otherwise to return a ``Failure" message. 
This {\it maximum likelihood} (ML) decoder is referred to as the standard ML decoder. 
Its formal definition for the binary asymmetric channel is given as follows:
\begin{definition}
For a code $C\subseteq \F_2^n$, the maximum likelihood decoder is the function $\mathcal{D}_C:\F_2^n \rightarrow C\cup \{\mathbf{f}\}$ defined by
\begin{equation*}
\mathcal{D}_C(\mathbf{y})=\left\{
\begin{array}{cl}
\mathbf{x}, & \text{ if $\mathbf{x}$ is the unique codeword that maximizes } \Pr(\mathbf{y}|\mathbf{x}),\\
\mathbf{f}, & \text{ otherwise,}
\end{array}\right.
\end{equation*}
where $\mathbf{f}\notin \F_2^n$ denotes a failure message.
 \end{definition} 

The analysis of error probability in maximum likelihood decoding is fundamental to the design of communication systems.  
The {\it pairwise error probability} $P_2({\bf x}\rightarrow {\bf x}')$ is defined as the probability that, under the ML decoding rule, ${\bf x}'$ is at least as likely as the transmitted codeword ${\bf x}$,
with ties counted as decoding failures:
\begin{equation}\label{eq:pairwise}
P_2(\mathbf{x}\rightarrow \mathbf{x}')=\sum_{\mathbf{y}\in V} \Pr(\mathbf{y}|\mathbf{x}),
\end{equation}
where $V=\{\mathbf{y}\in \F_2^n: \Pr({\bf y}|{\bf x}') \ge \Pr({\bf y}|{\bf x})\}.$

Assuming equiprobable transmission of codewords, i.e., $\Pr({\bf x})=1/|C|$ for every ${\bf x}\in C$, the {\it average error probability} of the ML decoder for the code $C$ is defined as
\begin{equation}\label{eq:eorr_prop}
\begin{aligned}
P_e(C)&=\frac{1}{|C|}\sum_{\mathbf{x}\in C}\sum_{\substack{\mathbf{y}\in \F_2^n \\ \mathcal{D}_C(\mathbf{y})\ne \mathbf{x}}}\Pr(\mathbf{y}|\mathbf{x})\\
&=1-\frac{1}{|C|}\sum_{\mathbf{x}\in C}\sum_{\substack{\mathbf{y}\in \F_2^n \\ \mathcal{D}_C(\mathbf{y})= \mathbf{x}}}\Pr(\mathbf{y}|\mathbf{x}).
\end{aligned}
\end{equation}

Exact computation of $P_e$  is generally intractable due to the exponential growth of decision regions with code length $n$. 
A standard analytical approach employs the union bound based on pairwise error probabilities:
\begin{align}\label{bound-PEP}	
P_e(C)\leq \frac{1}{|C|}\sum\limits_{{\bf x}\neq {\bf x}'\in C}P_2({\bf x}\rightarrow {\bf x}').
\end{align}

While the union bound in (\ref{bound-PEP}) reduces the problem to pairwise error probabilities $P_2({\bf x}\rightarrow {\bf x}')$, exact evaluation of $P_2$ remains nontrivial.  
In symmetric channels, $P_2$ depends on the Hamming distance; for linear codes, it is related to the weight distribution. 
This connection has motivated extensive research on the weight distributions of linear codes.

In asymmetric channels, however, $P_2$ depends separately on the directional distances $d_{01}$ and $d_{10}$.
Thus, the classical weight distribution and the minimum discrepancy parameters, 
which will be defined in the following subsection, do not retain enough information to aggregate the contributions depending on $d_{01}$ and $d_{10}$.
This motivates the development of tractable bounds that better capture directional asymmetry.
As shown later in Lemma \ref{P_2-comp}, once such directional information is retained, the pairwise error contributions can be expressed explicitly.

Several recent works have proposed discrepancy-based bounds on $P_e(C)$ for asymmetric channels. 
In the next subsection, we review these bounds.

\subsection{Discrepancy-based Bounds on $P_e(C)$}

Let $\gamma=\log_{\frac{q}{1-p}}(\frac{p}{1-q})$ denote the channel parameter introduced in \cite[Notation III.1]{GR2022}.
Using this parameter, the authors of \cite{GR2022} defined two discrepancy measures between binary vectors  $\mathbf{y},\mathbf{x}\in \F_2^n$: 
\begin{equation*}
\delta_\gamma(\mathbf{y},\mathbf{x})=\gamma d_{10}(\mathbf{y},\mathbf{x})+d_{01}(\mathbf{y},\mathbf{x}),
\end{equation*}
referred to as the \emph{discrepancy} between $\mathbf{y},\mathbf{x}$, and  
\begin{equation*}
\hat{\delta}_\gamma(\mathbf{y},\mathbf{x})=\delta_\gamma(\mathbf{y},\mathbf{x})-wt(\mathbf{y})(\gamma-1),
\end{equation*}
referred to as the \emph{symmetric discrepancy} between $\mathbf{y},\mathbf{x}$.
Here, as previously defined, $d_{10}({\bf y},{\bf x})=|\{i:y_i=1, x_i=0\}|$ and $d_{01}({\bf y},{\bf x})=|\{i:y_i=0, x_i=1\}|$. 
Therefore, in general, $d_{10}({\bf y},{\bf x})\ne d_{10}({\bf x},{\bf y})$ and $d_{01}({\bf y},{\bf x})\ne d_{01}({\bf x},{\bf y})$.
Consequently, the discrepancy $\delta_\gamma(\mathbf{y},\mathbf{x})$ is generally not symmetric when $\gamma >1$.
Furthermore, it is worth noting that, following the standard convention, we regard $\mathbf{x}$ as the channel input and $\mathbf{y}$ as the corresponding output.
Accordingly, in order to align the definition more directly with the assumption that $0\rightarrow 1$ is the less likely transition,
the parameter $\gamma$ is assigned to the coordinates counted in $d_{10}(\mathbf{y},\mathbf{x})$.

These give rise to two fundamental code parameters: the \emph{minimum discrepancy} $\delta_\gamma(C)$ and the \emph{minimum symmetric discrepancy} $\hat\delta_\gamma(C)$  of a binary code $C$, defined respectively as
\begin{equation*}
  \begin{aligned}
    \delta_\gamma(C) & =\min\big\{\delta_\gamma(\mathbf{c}_1,\mathbf{c}_2):\mathbf{c}_1,\mathbf{c}_2\in C,\mathbf{c}_1\neq \mathbf{c}_2\big\},\\
    \hat{\delta}_\gamma(C) & =\min\big\{\hat{\delta}_\gamma(\mathbf{c}_1,\mathbf{c}_2):\mathbf{c}_1,\mathbf{c}_2\in C,\mathbf{c}_1\neq \mathbf{c}_2\big\}.\\
  \end{aligned}
\end{equation*}

When combined with conventional weight distributions, these parameters lead to two incomparable upper bounds on the average error probability of codes.
\begin{lemma}\cite[Theorem IV.6]{GR2022}\label{lem:bound}
Keeping the notations above, let $C\subseteq \F_2^n$ be a code.
Then we have
\begin{equation*}
    P_e(C)\le  1-\frac{1}{|C|}\sum_{j=0}^n A_j \sum_{i=0}^n (1-q)^i (1-p)^{n-i} \cdot \sum_{s\in \mathcal{S}\left(\frac{\delta_\gamma(C)+(\gamma-1)(i-j)}{2}\right)}\left(\frac{q}{1-p}\right)^s \lambda(i,j,s),
\end{equation*}
and
\begin{equation*}
    P_e(C)\le 1-\frac{1}{|C|}\sum_{j=0}^n A_j \sum_{i=0}^n (1-q)^i (1-p)^{n-i} \cdot \sum_{s\in \mathcal{S}\left(\frac{\hat{\delta}_\gamma(C)+i(\gamma-1)}{2}\right)}\left(\frac{q}{1-p}\right)^s \lambda(i,j,s),
\end{equation*}
where $\lambda(i,j,s)$ denotes the number of binary vectors $\mathbf{y}\in \F_2^n$ of weight $i$ satisfying $\delta_\gamma(\mathbf{y},\mathbf{x})=s$ for a fixed codeword $\mathbf{x}\in C$ of weight $j$,
and $\mathcal{S}(h)=\{0\le s< h :s =a+\gamma b,~a,b \in \mathbb{N}\}$. 
\end{lemma}

However, these bounds have certain limitations: if two binary codes share identical weight distributions and the same minimum discrepancy $\delta_\gamma(C)$ (or minimum symmetric discrepancy $\hat{\delta}_\gamma(C)$), 
the resulting upper bounds on $P_e(C)$ cannot distinguish their performance, as illustrated in the following example.

\begin{exmpl}\label{exp:1}
Let $p=0.1,q=0.15$, and consider two binary codes  
\begin{equation*}
\begin{aligned}
  C_1 & =\{(1,1,1,0,0,0),(0,1,1,1,0,0),(1,1,0,0,0,0)\},\\
  C_2 & =\{(1,1,1,0,0,0),(0,0,0,1,1,1),(1,1,0,0,0,0)\}.
\end{aligned}
\end{equation*}
Both codes clearly have the same weight distribution $(0,0,1,2,0,0,0)$.
The channel parameter is $\gamma \approx 1.1944$, and we obtain
\begin{equation*}
  \begin{aligned}
    \delta_\gamma(C_1)=\delta_\gamma(C_2) & =1,\\
    \hat{\delta}_\gamma(C_1)=\hat{\delta}_\gamma(C_2) & = 3-2 \gamma \approx 0.6112.
  \end{aligned}
\end{equation*}
By Lemma \ref{lem:bound}, the same bound $0.5435$ (applying both $\delta_\gamma$ and $\hat{\delta}_\gamma$) holds for both $P_e(C_1)$ and $P_e(C_2)$.
Nevertheless, the actual error probabilities are $P_e(C_1)=0.2328$ and $P_e(C_2)=0.101$. 
\end{exmpl} 
Example \ref{exp:1} clearly demonstrates that the discrepancy-based bounds may fail to reflect the true performance difference between codes when their weight distributions and minimum discrepancy values coincide.

To address this limitation, we introduce in the next section the asymmetric Hamming bidistance for binary vectors, 
which offers a more refined  characterization of binary codes. 
By further incorporating its two-dimensional bidistance distribution, we then derive a more discriminative bound on the decoding error probability. 
Furthermore, we show analytically that our bound and the existing discrepancy-based bounds are generally incomparable, thereby clarifying their respective roles in performance estimation.

\section{A New Bound on the Average Error Probability Based on Bidistance Distribution} \label{sec:3}

Building on the limitations of existing discrepancy-based bounds discussed previously, 
this section introduces a refined analytical framework for estimating the average error probability of maximum-likelihood decoding over the binary asymmetric channel. 
We first define a two-dimensional distance measure, the asymmetric Hamming bidistance, which separately accounts for the $0\rightarrow1$ and $1\rightarrow0$ error directions, 
together with its associated bidistance distribution for a binary code. 
Using these constructs, we then derive a new upper bound on the decoding error probability that explicitly incorporates the directional asymmetry of the channel. 
Compared with the classical union bound, this bound can further improve the discriminability between codes that share identical conventional weight distributions and minimum discrepancy values.
 
\subsection{Asymmetric Hamming Bidistance and Its Distribution}

\begin{definition}\label{AHB}
Let ${\bf x}=(x_1,x_2,\ldots,x_n)$ and ${\bf y}=(y_1,y_2,\ldots,y_n)$ be two binary vectors in $\F_2^n$. We define the asymmetric Hamming bidistance (AHB) between ${\bf x}$ and ${\bf y}$ as the ordered pair 
\[d_A({\bf x},{\bf y})=(d_{10}({\bf x},{\bf y}),d_{01}({\bf x},{\bf y})),\]
where $d_{10}(\mathbf{x},\mathbf{y})=|\{i:x_i=1,y_i=0\}|,~d_{01}(\mathbf{x},\mathbf{y})=|\{i:x_i=0,y_i=1\}|.$

The two components count, respectively, the number of positions in which a $1$ in ${\bf x}$ corresponds to a $0$ in ${\bf y}$, 
and the number of positions in which a $0$ in ${\bf x}$ corresponds to a $1$ in ${\bf y}$.
\end{definition}
Obviously the conventional Hamming distance can be recovered as $d_H({\bf x},{\bf y})=d_{10}({\bf x},{\bf y})+d_{01}({\bf x},{\bf y})$.

For a binary code $C\subseteq\F_2^n$, we define its {\it bidistance distribution} as the two-dimensional array ${\mathbb A}_C$ with elements
\[A(d_{10},d_{01})=|\{({\bf x},{\bf y})\in C\times C:d_{10}(\mathbf{x},\mathbf{y})=d_{10},d_{01}(\mathbf{x},\mathbf{y})=d_{01}\}|,\]
for all $d_{10},d_{01}\in\{0,1,\ldots,n\}$ and $d_{10}+d_{01}\leq n$. 
Here, $A(d_{10},d_{01})$ is called the {\it frequency} of the ordered pair $(d_{10},d_{01})$. 
Clearly $A(0,0)=|C|$, so we generally omit this case.
If the bidistance distribution ${\mathbb A}_C$ is sparse, i.e., contains many zero entries, it can also be represented as a multiset 
\[{\cal A}_C=\{(d_{10},d_{01})^{A(d_{10},d_{01})}: 0<d_{10}+d_{01}\leq n\},\]
where the notation $(\cdot,\cdot)^j$  indicates that element $(\cdot,\cdot)$ appears $j$ times.

\begin{remark}
$(1)$ In general, $d_A({\bf x},{\bf y})\neq d_A({\bf y},{\bf x})$; indeed, $d_{10}({\bf x},{\bf y})= d_{01}({\bf y},{\bf x})$.

$(2)$ $A(i,j)=A(j,i)$ for any $i,j\geq0$ and $i+j\leq n$.

$(3)$ The bidistance distribution $\mathbb{A}_C$ refines the usual pairwise Hamming distance distribution. 
In particular, for linear codes, it can contain strictly more information than the classical weight
enumerator and therefore generally provides greater discriminative power in asymmetric settings.
For example, ${\mathbb A}_C$ contains sufficient information to derive both the asymmetric distance over the Z-channel \cite{CR1979} and the discrepancy $\delta_\gamma$.

$(4)$ Especially, for a linear code $C$, its conventional weight distribution $\{A_i\}_{i=0}^n$ is obtained by summing over the antidiagonals:
\[A_i=\frac{1}{|C|}\sum\limits_{\substack{d_{10},d_{01}\geq0\\ d_{10}+d_{01}=i}} A(d_{10},d_{01}).\]
\end{remark}

We now illustrate the preceding definitions with a concrete example.
\begin{exmpl}\label{exp:1-1}
Let  $C_1$ and $C_2$ be the binary codes in Example \ref{exp:1}.
Their bidistance distributions  are shown as follows: 
$$\setlength{\arraycolsep}{2pt}
\mathbb{A}_{C_1}=\begin{pmatrix}
3 & 1 & 0 & 0 & 0 & 0 & 0 \\
1 & 2 & 1 & 0 & 0 & 0 & 0 \\
0 & 1 & 0 & 0 & 0 & 0 & 0 \\
0 & 0 & 0 & 0 & 0 & 0 & 0 \\
0 & 0 & 0 & 0 & 0 & 0 & 0 \\
0 & 0 & 0 & 0 & 0 & 0 & 0 \\
0 & 0 & 0 & 0 & 0 & 0 & 0
\end{pmatrix},~~
\mathbb{A}_{C_2}=\begin{pmatrix}
3 & 1 & 0 & 0 & 0 & 0 & 0 \\
1 & 0 & 0 & 0 & 0 & 0 & 0 \\
0 & 0 & 0 & 1 & 0 & 0 & 0 \\
0 & 0 & 1 & 2 & 0 & 0 & 0 \\
0 & 0 & 0 & 0 & 0 & 0 & 0 \\
0 & 0 & 0 & 0 & 0 & 0 & 0 \\
0 & 0 & 0 & 0 & 0 & 0 & 0
\end{pmatrix}.$$
Equivalently, we can use the following multisets for a simplified representation:
\begin{equation*}
\begin{aligned}
{\cal A}_{C_1}=\{(0, 1), (1, 0), (1, 1)^{ 2}, (1, 2), (2, 1)\},\\
{\cal A}_{C_2}=\{(0, 1), (1, 0), (2, 3), (3, 2), (3, 3)^{ 2}\}.
\end{aligned}
\end{equation*}
\end{exmpl}

\subsection{The New Union Bound on the Average Error Probability}

In this subsection, we first establish that the pairwise error probability $P_2({\bf x}\rightarrow {\bf x}')$ depends explicitly on the two components of the asymmetric Hamming bidistance $d_A({\bf x},{\bf x}')=(d_{10},d_{01})$. 
By incorporating this directional information, we can refine the classical union bound in terms of the bidistance distribution to obtain a more discriminative estimate of the average decoding error probability $P_e(C)$.

For a binary asymmetric channel with transition probabilities $(p,q)$, recall that the {\it log-likelihood ratio} (LLR) between two codewords ${\bf x}$ and ${\bf x}'$ given the received vector ${\bf y}$ is defined as
\begin{equation*}
LLR_{{\bf x},{\bf x}'}(\mathbf{y})=\log\left(\frac{\Pr(\mathbf{y}|\mathbf{x})}{\Pr(\mathbf{y}|\mathbf{x}')}\right),
\end{equation*}
which factorizes as
\begin{equation*}
LLR_{{\bf x},{\bf x}'}(\mathbf{y})=\sum\limits_{k=1}^n\log\left(\frac{\Pr(y_k|x_k)}{\Pr(y_k|x_k')}\right),
\end{equation*}
for a discrete memoryless channel.
(The base of the logarithm here can be $2$ or any other real number greater than $1$.)
Therefore, the pairwise error probability can be expressed as
\[P_2({\bf x}\rightarrow {\bf x}')=\sum_{LLR_{{\bf x},{\bf x}'}({\bf y})\leq 0}\Pr(\mathbf{y}|\mathbf{x}).\]

Let $d_A({\bf x},{\bf x}')=(d_{10},d_{01})$. Define $k_{10}$ as the number of positions, among the $d_{10}$ positions where $x_i=1$ and $x_i'=0$, 
in which the channel actually flips the transmitted $1$ to a received $0$. 
Similarly, define $k_{01}$ as the number of positions, among the $d_{01}$ positions, where $x_i=0$ and $x_i'=1$, in which the channel flips the transmitted $0$ to a received $1$. 
Formally, 
\begin{equation}\label{k_{01}}
	k_{10}=|\{i:x_i=1,x_i'=0,y_i=0\}|,~k_{01}=|\{i:x_i=0,x_i'=1,y_i=1\}|.
\end{equation}

Straightforward algebra shows the following result.
\begin{lemma}\label{equiv}
For a binary asymmetric channel with transition probabilities $p,q~(0<p\le q<1/2)$, the $LLR_{{\bf x},{\bf x}'}(\mathbf{y})\leq 0$ if and only if $k_{10}+k_{01}\geq \lceil\frac{d_{10}\gamma+d_{01}}{\gamma+1}\rceil$,
where $\gamma=\log_{\frac{q}{1-p}}(\frac{p}{1-q})$ is the channel parameter and $\lceil\cdot\rceil$ denotes the ceiling function.
\end{lemma}
\begin{proof}
In the binary asymmetric channel, only the positions where the two codewords differ contribute to the LLR, that is, the positions in Regions $R2$ and $R3$.
The position coordinates corresponding to $k_{10}$ and $k_{01}$ also lie in these regions by their definition.
\[
\begin{array}{c}
\mathbf{x}\\
\mathbf{x}'\\
\mathbf{y}
\end{array}
\overbrace{
\begin{array}{cccc}
  0 & 0 & 1 & 1\\
  0 & 0 & 1 & 1 \\
  0 & 1 & 0 & 1
\end{array}}^{R1}
\rule[-7.5mm]{0.6pt}{18mm} 
\overbrace{
\begin{array}{cc}
    1 & 1 \\
    0 & 0 \\
    0 & 1
\end{array}
}^{R2}
\underbrace{
\begin{array}{cc}
    0 & 0 \\
    1 & 1 \\
    0 & 1
\end{array}
}_{R3}
\]
For a single position, the likelihood ratio satisfies
\[
\frac{P(y_k\mid 0)}{P(y_k\mid 1)}=
\begin{cases}
    \dfrac{1-p}{q}, & y_k=0,\\[6pt]
    \dfrac{p}{1-q}, & y_k=1 .
\end{cases}
\]
Recall that \(d_{10}=d_{10}(\mathbf{x},\mathbf{x}')\), \(d_{01}=d_{01}(\mathbf{x},\mathbf{x}')\), and let $k_{10}$ and $k_{01}$ be the random variables defined in (\ref{k_{01}}).
Then  
\begin{align*}
	{\rm LLR}_{{\bf x},{\bf x}'}(\mathbf{y})&=k_{01}\log\frac{p}{1-q}+(d_{01}-k_{01})\log\frac{1-p}{q}+k_{10}\log\frac{q}{1-p}+(d_{10}-k_{10})\log\frac{1-q}{p}\\
	&=(k_{01}+k_{10}-d_{10})\log\frac{p}{1-q}+(k_{01}+k_{10}-d_{01})\log\frac{q}{1-p}\\
	&=\big[(k_{01}+k_{10}-d_{10})\gamma+(k_{01}+k_{10}-d_{01})\big]\cdot \log\frac{q}{1-p}.
\end{align*}
Because \(0<p\le q<1/2\) implies \(\frac{q}{1-p}<1\), we have \(\log\frac{q}{1-p}<0\). 
Consequently, 
\[
\operatorname{LLR}_{\mathbf{x},\mathbf{x}'}(\mathbf{y})\le 0
\Longleftrightarrow
(k_{01}+k_{10}-d_{10})\gamma+(k_{01}+k_{10}-d_{01})\geq0.
\]
Since $k_{01}+k_{10}$ takes only integer values, rearranging the inequality yields the desired condition, thereby completing the proof.
\end{proof}

Due to the memoryless property of the channel and the fact that the position sets corresponding to $d_{01}$ and $d_{10}$ are disjoint, 
the random variables $k_{01} \sim \operatorname{Bin}(d_{01}, p)$ and $k_{10} \sim \operatorname{Bin}(d_{10}, q)$ are independent, 
each counting the actual directional errors in the respective differing positions of the two codewords. 
Applying the equivalence established in Lemma \ref{equiv}, we derive the following closed-form expression for the pairwise error probability.

\begin{lemma}\label{P_2-comp}
Define ${\cal R}_{(d_{10},d_{01})}=\{(i,j):0\leq i\leq d_{10},0\leq j\leq d_{01},i+j\geq \lceil(d_{10}\gamma+d_{01})/(\gamma+1)\rceil\}$. 
Then
\begin{equation}\label{P_2-form}
P_2(\mathbf{x}\rightarrow \mathbf{x}')=\sum_{(i,j)\in{\cal R}_{(d_{10},d_{01})}}\binom{d_{10}}{i}q^{i}(1-q)^{d_{10}-i}\cdot\binom{d_{01}}{j}p^{j}(1-p)^{d_{01}-j}.
\end{equation}
\end{lemma}

\begin{proof}
For fixed $\mathbf{x}$ and $\mathbf{x}'$, let $V=\{\mathbf{y}\in\mathbb{F}_2^n: \operatorname{LLR}_{\mathbf{x},\mathbf{x}'}(\mathbf{y})\leq 0\}.$
We adopt the partition into Regions $R1$, $R2$, and $R3$ introduced in the proof of Lemma~\ref{equiv}. 
In Region $R1$, the two codewords have identical coordinates, i.e., $x_k=x'_k$. 
Therefore, regardless of the value of $y_k$,
\[\log\frac{\Pr(y_k\mid x_k)}{\Pr(y_k\mid x'_k)}=0,~~\text{for every }k\in R1.\] 
Hence, whether a received vector $\mathbf{y}$ belongs to $V$ depends only on its coordinates in $R2\cup R3$.

Let $V'=\{\mathbf{y}|_{R2\cup R3}:\mathbf{y}\in V\}$ denote the set of restrictions of the vectors in $V$ to the coordinates in $R2\cup R3$. 
By the memoryless property of the channel,
\begin{align*}
P_2(\mathbf{x}\rightarrow\mathbf{x}')
&=\sum_{\mathbf{y}\in V}\Pr(\mathbf{y}\mid\mathbf{x})\\
&=\sum_{\mathbf{y}|_{R2\cup R3}\in V'} \prod_{k\in R2\cup R3}\Pr(y_k\mid x_k) \sum_{(y_k)_{k\in R1}} \prod_{k\in R1}\Pr(y_k\mid x_k).
\end{align*}
For the coordinates in $R1$, the received symbols may take all possible values in $\mathbb{F}_2$, since these coordinates do not affect the likelihood-ratio condition. 
Then the contribution from the coordinates in $R1$ can be marginalized as follows:
\begin{align*}
\sum_{(y_k)_{k\in R1}\in\mathbb{F}_2^{|R1|}} \prod_{k\in R_1}\Pr(y_k\mid x_k) 
&= \prod_{k\in R1} \left(\sum_{y_k\in\mathbb{F}_2}\Pr(y_k\mid x_k) \right)\\
&= \prod_{k\in R1} \left(\Pr(0\mid x_k)+\Pr(1\mid x_k)\right)\\
&=1,
\end{align*}
where $|R1|$ is the number of coordinates in $R1$.
Consequently,
\begin{align*} P_2(\mathbf{x}\rightarrow\mathbf{x}') 
&= \sum_{\mathbf{y}|_{R2\cup R3}\in V'} \prod_{k\in R2\cup R3}\Pr(y_k\mid x_k)\\
&= \sum_{{\rm LLR}_{{\bf x},{\bf x}'}({\bf y})\leq 0} \prod_{k\in R2\cup R3}\Pr(y_k\mid x_k),
\end{align*}
where the second equality follows from the fact that the coordinates in $R1$ do not contribute to the likelihood ratio, 
and hence the condition $\mathbf{y}|_{R2\cup R3}\in V'$ is also equivalent to $\operatorname{LLR}_{\mathbf{x},\mathbf{x}'}(\mathbf y)\le0$.

Based on Lemma \ref{equiv} and the independence of $k_{10}$ and $k_{01}$, 
the pairwise error probability can be written as a double summation over the values of $k_{10}$ and $k_{01}$ that satisfy the inequality condition. 
Specifically, 
\[P_2({\bf x}\to{\bf x}')=\Pr\big(k_{10}+k_{01}\geq \lceil\tau\rceil\big),~\tau=\frac{d_{10}\gamma+d_{01}}{\gamma+1}.\]
Since $k_{10}$ and $k_{01}$ are independent binomial random variables, 
\begin{align*}
P_2({\bf x}\to{\bf x}')&=\sum_{(i,j)\in{\cal R}_{(d_{10},d_{01})}}\Pr(k_{10}=i)\cdot \Pr(k_{01}=j)\\
&=\sum_{(i,j)\in{\cal R}_{(d_{10},d_{01})}}\binom{d_{10}}{i}q^{i}(1-q)^{d_{10}-i}\cdot\binom{d_{01}}{j}p^{j}(1-p)^{d_{01}-j},
\end{align*}
thereby completing the proof.
\end{proof}

This explicit form in Equation (\ref{P_2-form}) demonstrates that, for fixed channel transition probabilities $p$ and $q$ (embedded in $\gamma$), 
the pairwise error probability $P_2({\bf x}\to{\bf x}')$ depends solely on the numerical values of the directional distances $d_{10}$ and $d_{01}$, 
and not on the specific realizations of the codewords ${\bf x}$ and ${\bf x}'$. 
Therefore, in the next theorem, we will aggregate such pairwise contributions through the bidistance distribution to derive a more discriminative bound on the average error probability $P_e(C)$. 

\begin{theorem}\label{Refined-bound}
In a binary asymmetric channel with transition probabilities $p,q~(0<p\le q<1/2)$, define $\gamma=\log_{\frac{q}{1-p}}(\frac{p}{1-q})$.
Let $C\subseteq \F_2^n$ be a binary code and ${\cal H}_C=\{d_A(\mathbf{x},\mathbf{y}):\mathbf{x},\mathbf{y}\in C, \mathbf{x}\ne \mathbf{y}\}$.  
Then we have
\begin{equation}\label{eq:bound}
P_e(C)\le  \frac{1}{|C|}\sum_{(d_{10},d_{01})\in \mathcal{H}_C} A(d_{10},d_{01}) P(d_{10},d_{01}),
\end{equation}
where 
\begin{equation*}
P(d_{10},d_{01})=\sum_{(i,j)\in{\cal R}_{(d_{10},d_{01})}}\binom{d_{10}}{i}q^{i}(1-q)^{d_{10}-i}\cdot\binom{d_{01}}{j}p^{j}(1-p)^{d_{01}-j},
\end{equation*}
and $${\cal R}_{(d_{10},d_{01})}=\{(i,j):0\leq i\leq d_{10},0\leq j\leq d_{01},i+j\geq \lceil(d_{10}\gamma+d_{01})/(\gamma+1)\rceil\}.$$
\end{theorem}
\begin{proof}
Applying the result of Lemma \ref{P_2-comp} to the union bound in (\ref{bound-PEP}) yields the desired conclusion.	
\end{proof}

\begin{exmpl}\label{exp:3}
Let $C_1$ and $C_2$ be the two binary codes described in Example \ref{exp:1}. 
Applying Theorem \ref{Refined-bound} yields the refined upper bounds 
\begin{equation*}
P_e(C_1)\le 0.2683,~~P_e(C_2)\le 0.1129,
\end{equation*}
which are substantially closer to the true error probabilities $P_e(C_1)=0.2328$ and $P_e(C_2)=0.101$, than the discrepancy-based bounds (both equal to $0.5435$).
\end{exmpl}

Example \ref{exp:3} demonstrates the advantage of our bound when the two discrepancy-based bounds coincide (and fail to distinguish codes). 
In the following example, we further compare our bound with the two discrepancy-based bounds in a case where the latter yield distinct values. 
For a clear and fair comparison, we adopt exactly the same binary codes used in \cite{GR2022} (namely, \cite[Example IV.7 and Example IV.9]{GR2022}).

\begin{exmpl}\label{exp:4}
Let $p=0.1$, $C_1=\{(0,1,0,1,1),(1,1,0,0,0),(1,0,1,1,1)\}$ and $C_2=\{(1,1,1,1),(1,0,0,1),(0,0,0,0)\}$. 
Figures \ref{fig:1} and \ref{fig:2} depict how the two discrepancy-based bounds (Lemma \ref{lem:bound}) and our bound (Theorem \ref{Refined-bound}) vary as $q$ ranges from $0.1$ to $0.49$. 
The figures also illustrate the relationship between these bounds and the exact average error probabilities (shown as green curves), which are obtained by exhaustive enumeration.
\begin{figure}[!htbp]
\centering
\begin{minipage}[b]{.4\linewidth}
\includegraphics[width=\linewidth]{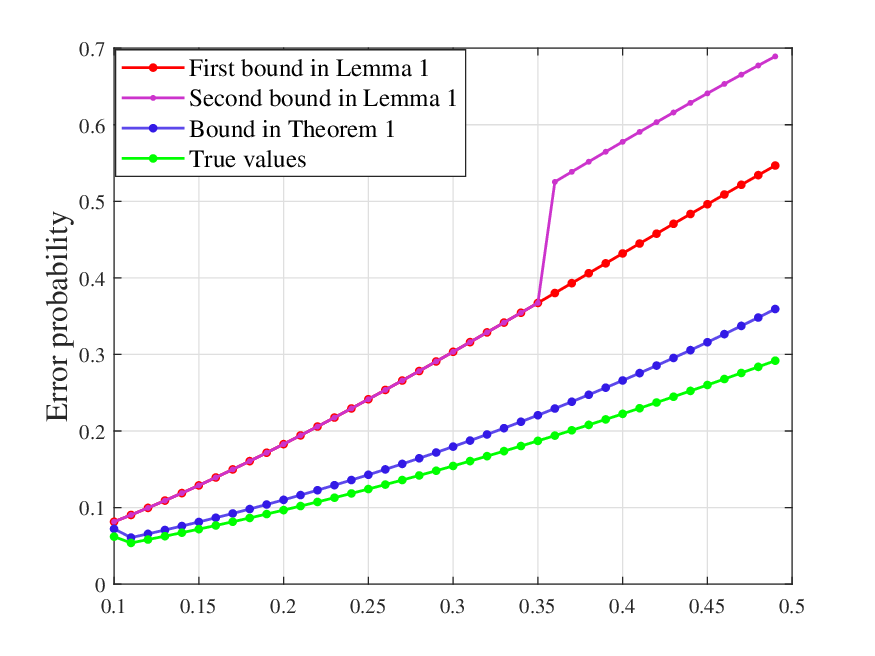}
\captionsetup{font={scriptsize}}
\caption{Values of the three bounds and true values of $P_e(C_1)$ in \\ Example \ref{exp:4}}
\label{fig:1}
\end{minipage}%
\begin{minipage}[b]{.4\linewidth}
\includegraphics[width=\linewidth]{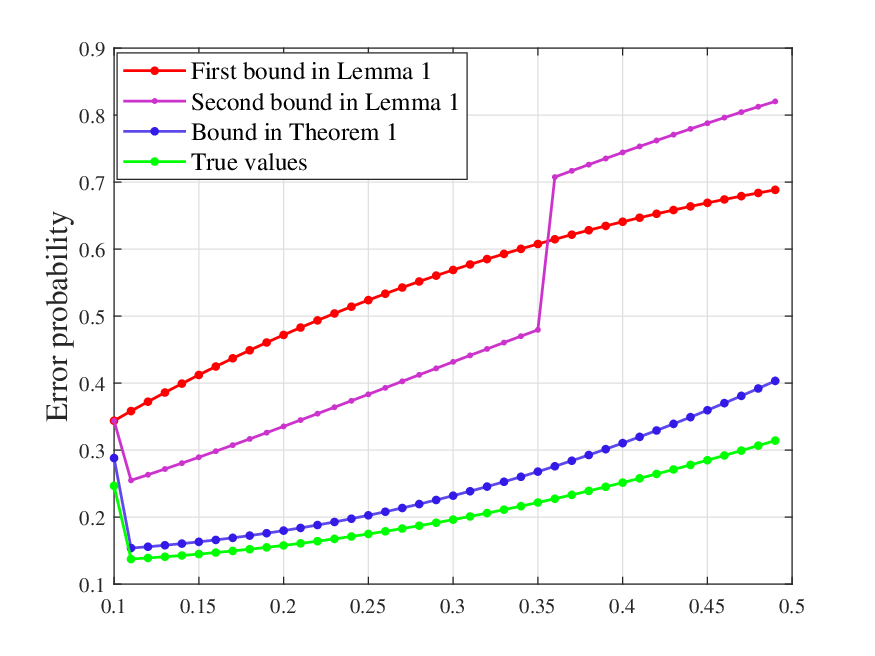}
\captionsetup{font={scriptsize}}
\caption{Values of the three bounds and true values of $P_e(C_2)$ in \\ Example \ref{exp:4}}
\label{fig:2}
\end{minipage}%
\end{figure}
\end{exmpl}

Although in the preceding examples our bound is closer to the true $P_e$ than the two discrepancy-based bounds, the three bounds are generally incomparable, as illustrated in the following example.

\begin{exmpl}\label{exp:5}
Let $ C = \{(1,1,1,0),(0,1,1,1),(1,1,0,0)\} $. 
Figures \ref{fig:3} -- \ref{fig:4} illustrate how the two discrepancy-based bounds (Lemma \ref{lem:bound}) and our bound (Theorem \ref{Refined-bound}) change with $ q $ for values of $ p $ at $ 0.1, 0.2, 0.3, $ and $ 0.4 $. 
Throughout this paper, we assume that $ 0 < p \leq q < 1/2 $. 
Therefore, only the sections of the curves where $ q \geq p $ are displayed.
\begin{figure}[!htbp]
\centering
\begin{minipage}[b]{.24\linewidth}
\includegraphics[width=\linewidth]{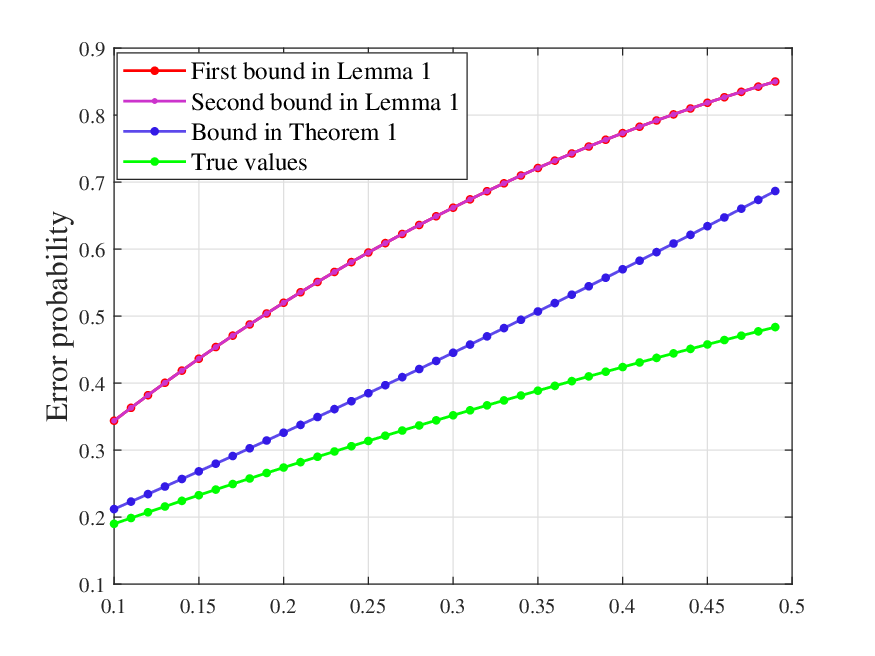}
\captionsetup{font={scriptsize}}
\caption{Values of the three bounds and \\ true values of $P_e(C)$ in Example \ref{exp:5} for\\ $p=0.1$}
\label{fig:3}
\end{minipage}%
\begin{minipage}[b]{.24\linewidth}
\includegraphics[width=\linewidth]{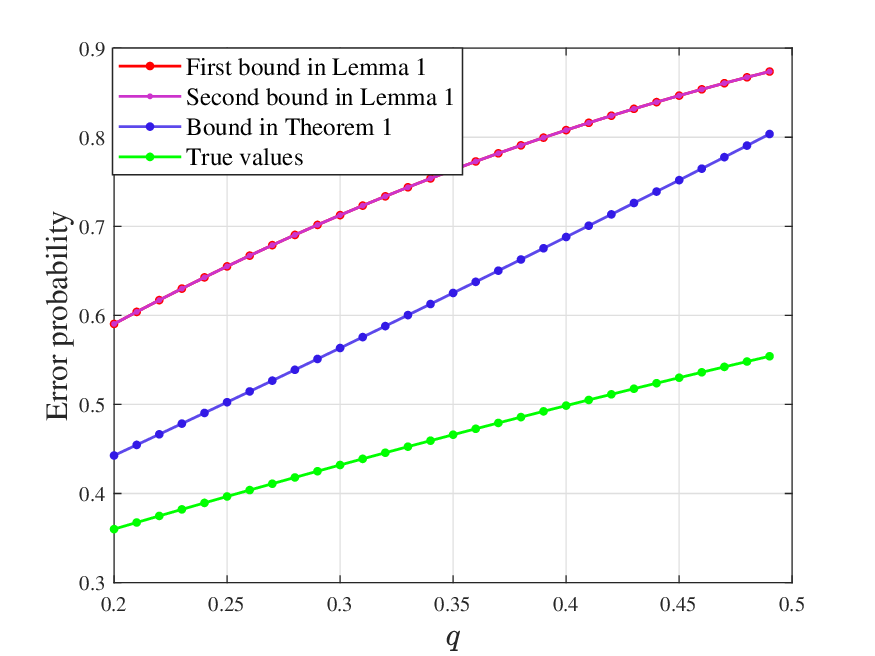}
\captionsetup{font={scriptsize}}
\caption{Values of the three bounds and \\ true values of $P_e(C)$ in Example \ref{exp:5} for\\ $p=0.2$}
\label{fig:2-2}
\end{minipage}%
\begin{minipage}[b]{.24\linewidth}
\includegraphics[width=\linewidth]{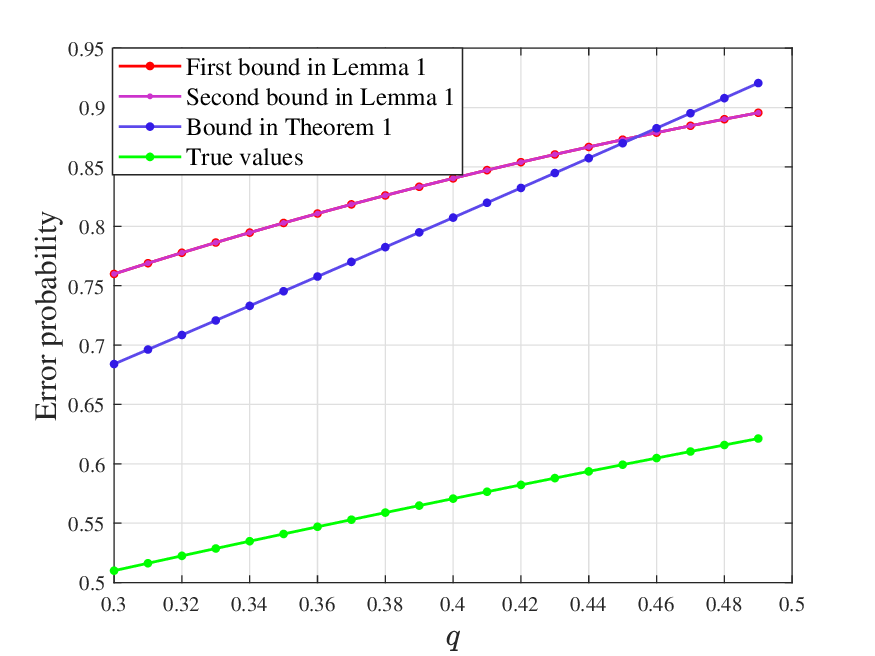}
\captionsetup{font={scriptsize}}
\caption{Values of the three bounds and \\ true values of $P_e(C)$ in Example \ref{exp:5} for\\ $p=0.3$}
\label{fig:3-2}
\end{minipage}%
\begin{minipage}[b]{.24\linewidth}
\includegraphics[width=\linewidth]{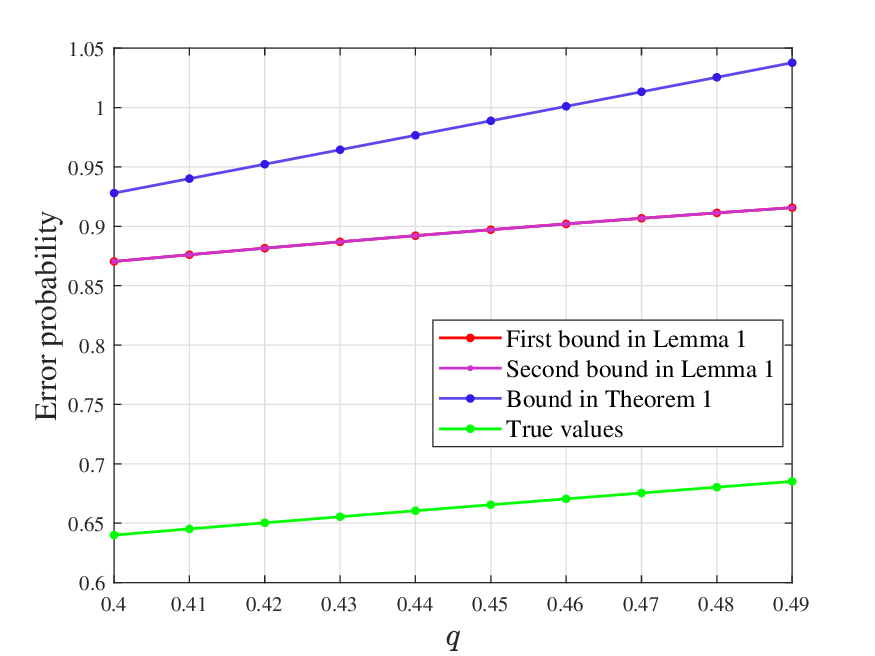}
\captionsetup{font={scriptsize}}
\caption{Values of the three bounds and \\  true values of $P_e(C)$ in Example \ref{exp:5} for\\ $p=0.4$}
\label{fig:4}
\end{minipage}
\end{figure}

For this code, the two discrepancy-based bounds in Lemma \ref{lem:bound} coincide. 
Denote their common value by $B_{[7]}(p,q)$, and denote the bound in Theorem \ref{Refined-bound} by $B_{\rm AHB}(p,q)$. 
A direct calculation gives
$$B_{[7]}(p,q)=1-\frac{(1-p)^2(1-q)^2+2(1-p)(1-q)^3}{3},$$
whereas
$$B_{\rm AHB}(p,q)=p+q+\frac{(p+q)^2}{3}-\frac{2pq(p+q)}{3}.$$
Hence, the crossover between the two bounds is characterized by the implicit curve
$$B_{[7]}(p,q)=B_{\rm AHB}(p,q).$$
The behavior of this curve within the admissible region $0<p\leq q<1/2$ can be described more precisely.
The crossover curve meets the boundary $q=1/2$ when
\[p=p_1\approx 0.266281,\]
and meets the boundary $q=p$ when 
\[p=p_2\approx0.361172,\]
where $p_2$ is the unique root in $(0,1/2)$ of $3p^3-16p^2+22p-6=0.$

For $0<p\lesssim p_1$, the AHB-based bound is smaller than the common discrepancy-based bound for every admissible $q$. 
For $p_1\lesssim p\lesssim p_2$, there is a unique crossover value $q^*(p)$ such that
\[B_{\mathrm{AHB}}(p,q)<B_{[7]}(p,q),~~p\leq q<q^*(p),\]
while
\[B_{[7]}(p,q)<B_{\mathrm{AHB}}(p,q),~~ q^*(p)<q<\frac{1}{2}.\]
For $p\gtrsim p_2$, the discrepancy-based bound is smaller throughout the admissible range.

These three regimes are illustrated by Figures \ref{fig:3} -- \ref{fig:4}. 
For $p=0.1$ and $p=0.2$, as shown in Figures \ref{fig:3} and \ref{fig:2-2}, respectively, $p<p_1$, 
and the AHB-based bound is smaller throughout the admissible range of $q$. 
For $p=0.3$, which lies between $p_1$ and $p_2$, 
Figure \ref{fig:3-2} exhibits a crossover at \[q^*(0.3)\approx0.454459.\]
Thus, the AHB-based bound is sharper for $0.3\leq q<0.454459$, whereas the discrepancy-based bound becomes sharper for $q>0.454459$. 
Finally, for $p=0.4>p_2$, Figure \ref{fig:4} shows that the discrepancy-based bound is smaller throughout the admissible region $q\geq 0.4$.
\end{exmpl}

Example \ref{exp:5} shows that the bounds based on AHB and those based on discrepancy are generally not comparable. 
Their relative ranking depends both on the structure of the code and the parameters of the channel.

We have now introduced the concepts of asymmetric bidistance and its distribution for binary codes, and derived a more discriminative bound on the average error probability based on these notions. 
However, determining the exact bidistance distribution of a general binary code is considerably more challenging than computing its conventional weight distribution. 
In the remainder of this paper, we address this challenge for several specific classes of binary codes by employing combinatorial constructions, thereby obtaining their complete bidistance distributions.

\section{Asymmetric Bidistance Distribution for Few-Weight Codes}\label{sec:4}

For binary vectors $\mathbf{x},\mathbf{y}\in \F_2^n \backslash\{\mathbf{0}\}$ and $\mathbf{x}\ne \mathbf{y}$, the definition of the asymmetric Hamming bidistance immediately yields the following relations:
\begin{equation}\label{eq:system}
\left\{
\begin{split}
   d_{10}(\mathbf{x},\mathbf{y})+d_{01}(\mathbf{x},\mathbf{y}) & =wt(\mathbf{x}+\mathbf{y}), \\
   d_{10}(\mathbf{x},\mathbf{y})+d_{11}(\mathbf{x},\mathbf{y}) & =wt(\mathbf{x}), \\
   d_{01}(\mathbf{x},\mathbf{y})+d_{11}(\mathbf{x},\mathbf{y}) & =wt(\mathbf{y}).
\end{split}
\right.
\end{equation}
Solving this linear system gives explicit expressions for the directional distances in terms of the weights of ${\bf x},{\bf y}$ and their sum:
\begin{equation}\label{eq:system2}
\left\{
\begin{split}
   d_{10}(\mathbf{x},\mathbf{y}) & =\frac{1}{2} \big(wt(\mathbf{x}+\mathbf{y})+wt(\mathbf{x})-wt(\mathbf{y})\big), \\
   d_{01}(\mathbf{x},\mathbf{y}) & =\frac{1}{2} \big(wt(\mathbf{x}+\mathbf{y})-wt(\mathbf{x})+wt(\mathbf{y})\big).
\end{split}
\right.
\end{equation}
From (\ref{eq:system2}), it follows that once the weights of ${\cal C}$ and ${\cal C}+{\cal C}$ (defined as $\{\mathbf{x}+\mathbf{y}: \mathbf{x},\mathbf{y}\in {\cal C} \}$) are known,  the possible values of the asymmetric bidistance between any two distinct codewords in ${\cal C}$ are determined. 
However, determining how often each admissible pair actually occurs -- i.e., obtaining the full bidistance distribution -- remains a challenging task in general, even for linear codes (for which ${\cal C}+{\cal C}={\cal C}$).

In this section, we address this challenge for specific families of two-weight and three-weight codes. 
By employing two combinatorial structures -- strongly regular graphs and 3-class association schemes -- we can compute their complete asymmetric bidistance distributions. 
These structures will be treated separately in the following two subsections.

\subsection{Two-weight Codes via Strongly Regular Graphs}\label{sec:4-1}

In this subsection, we employ strongly regular graphs to compute the bidistance distribution of two-weight projective codes (excluding the zero codeword). 
We begin by recalling the necessary definitions and known results.

A simple graph $G$ with $v$ vertices is called \emph{strongly regular} with parameters $(v,K,\lambda,\mu)$ if
it is regular of valency $K$, and for any two distinct vertices, the number of common neighbors is $\lambda$ or $\mu$ according as the two vertices are adjacent or non-adjacent.
If $0<K<v-1$, the graph is said to be \emph{nontrivial}. 
We mainly consider nontrivial strongly regular graphs in this paper.
The \emph{complement} of $G$ is also strongly regular, with parameters $(v,\bar{K},\bar{\lambda},\bar{\mu})$,
where $\bar{K}=v-K-1,~\bar{\lambda}=v-2K+\mu-2$ and $\bar{\mu}=v-2K+\lambda$.

Two equivalent definitions of the adjacency matrix are commonly used in the literature. 
The first, denoted by $B_G=(b_{ij})^{v}_{i,j=1}$ and introduced in \cite{CK1986}, sets $b_{ij}=1$ for distinct adjacent vertices and  $b_{ij}=0$ otherwise.
This matrix satisfies:
\begin{equation}\label{eq:mat-B1}
B_G J_v=J_v B_G=K J_v,
\end{equation}
\begin{equation}\label{eq:mat-B2}
B_G^2-(\lambda-\mu)B_G-(K-\mu)I_v=\mu J_v,
\end{equation}
where $J_v$ denotes the all-ones matrix of order $v$  and $I_v$ is the identity matrix.

The second definition, due to Delsarte\cite{Delsarte1972}, is given by $A_G=(a_{ij})^{v}_{i,j=1}$ with
 $a_{ii}=0$, and for  $i\neq j$, $a_{ij}=a_{ji}=-1$ if the vertices $i$ and $j$ are adjacent, and $a_{ij}=a_{ji}=1$ otherwise.
The matrix satisfies:
\begin{equation}\label{eq:mat-A1}
A_G J_v=J_v A_G=\rho_0 J_v,
\end{equation}
\begin{equation}\label{eq:mat-A2}
(A_G-\rho_1 I_v)(A_G-\rho_2 I_v)=(v-1+\rho_1 \rho_2)J_v,
\end{equation}
where the constants $\rho_0$, $\rho_1$, $\rho_2$ are determined by the graph parameters as follows.

\begin{lemma}\label{lem:p4}
Let $G$ be a nontrivial strongly regular graph with parameters $(v,K,\lambda,\mu)$. 
Then we have $\rho_0=v-1-2K$, and $\rho_1,\rho_2$ satisfy
\begin{equation}\label{eq:coe}
\left\{
\begin{array}{l}
2(\mu-\lambda-1)=\rho_1+\rho_2, \\
2\lambda+2\mu-4K+1=\rho_1 \rho_2.
\end{array}
\right.
\end{equation}   
\end{lemma}
\begin{proof}
From the two definitions of the adjacency matrices, we obtain the fundamental relation
\begin{equation}\label{eq:relation}
A_G=J_v-I_v-2B_G.
\end{equation}

Substituting (\ref{eq:relation}) into (\ref{eq:mat-A1}) and using (\ref{eq:mat-B1}) yields
\begin{equation}\label{eq:K}
\begin{split}
    (J_v-I_v-2B_G)J_v &=\rho_0 J_v,\\
    J_v^2-J_v-2KJ_v &=\rho_0 J_v,\\
    (v-\rho_0-1)J_v &=2KJ_v,
\end{split}
\end{equation}
which gives $\rho_0=v-1-2K$. 
Next, substituting (\ref{eq:relation}) into 
(\ref{eq:mat-B2}) gives
\begin{equation*}
\frac{1}{4}(J_v-I_v-A_G)^2-\frac{1}{2}(\lambda-\mu)(J_v-I_v-A_G)-(K-\mu)I_v=\mu J_v.
\end{equation*}
Simplifying the expression and employing (\ref{eq:mat-A1}) leads to
\begin{equation}\label{eq:15}
   A_G^2+2(\lambda-\mu+1)A_G+(v-2\rho_0-2\lambda-2\mu-2)J_v +(2\lambda+2\mu-4K+1)I_v=0.
\end{equation}
On the other hand, expanding (\ref{eq:mat-A2}) yields
\begin{equation}\label{eq:nag}
A_G^2-(\rho_1+\rho_2)A_G-(v-1+\rho_1 \rho_2)J_v+\rho_1 \rho_2 I_v =0.
\end{equation}
Since $A_G,J_v$ and $I_v$ are linearly independent, comparing the coefficients of (\ref{eq:15}) and (\ref{eq:nag}) directly yields the system of equations in the lemma.
\end{proof}

The connection between binary two-weight projective codes and strongly regular graphs was established by Delsarte \cite{Delsarte1972}.
Let $C$ be a two-weight $[n,k]$ linear code over $\F_2$ with nonzero weights $w_1<w_2$. 
The graph $G$ associated with $C$ is defined as follows:
\begin{itemize}
  \item the vertex set $V(G)$ consists of all codewords of $C$;
  \item the edge set is defined as $E(G)=\bigl\{(\mathbf x,\mathbf y):\mathbf x,\mathbf y\in C,\ \mathbf x\ne\mathbf y \text{ and } d_H(\mathbf x,\mathbf y)=w_1\bigr\}.$
\end{itemize}  
The following result provides the explicit parameters of the resulting strongly regular graph.

\begin{lemma}\cite[Theorem 2]{Delsarte1972}\label{lem:D1972}
Let $C$ be a two-weight $[n,k]$ projective code over $\F_2$ with weights $w_1<w_2$. 
Then the associated graph $G$ is strongly regular, 
and the constants $\rho_0,\rho_1,\rho_2$ appearing in (\ref{eq:mat-A1}) and (\ref{eq:mat-A2}) are given by
\begin{equation*}
\begin{split}
   (w_2-w_1)\rho_0 &= n\cdot 2^k-(w_1+w_2)(2^k-1), \\
   (w_2-w_1)\rho_i &= w_1+w_2-(1+(-1)^i)\cdot 2^{k-1},~~i=1,2.
\end{split}
\end{equation*}
\end{lemma}

Solving the system (\ref{eq:coe}) together with the expressions for $\rho_0,\rho_1,\rho_2$ provided in Lemma \ref{lem:D1972} yields explicit formulas for the parameters $K,\lambda,\mu$ of the associated strongly regular graph $G$. 
These are summarized in the following lemma.

\begin{lemma}\label{lem:values}
Let $C$ be a two-weight $[n,k]$ projective code over $\F_2$ with nonzero weights $w_1<w_2$, and let 
$G$ be the strongly regular graph  associated with $C$ as constructed above. 
Then $G$ has parameters $(v,K,\lambda,\mu)$, where $v=2^k$,
\begin{equation*}
K = \frac{(2^{k}-1)w_2-2^{k-1}n}{(w_2-w_1)},
\end{equation*}
\begin{equation*}
  \begin{split}
    \lambda &= \frac{(2^{k}-2)w_2^2+(3-2^{k})w_1w_2+2^{k-1}((n-1)w_1-nw_2)}{(w_2-w_1)^2}, \\
    \mu &= \frac{2^{k}w_2^2+(1-2^k)w_1w_2-2^{k-1}((n+1)w_2-nw_1)}{(w_2-w_1)^2}.
  \end{split}
\end{equation*}
\end{lemma}
\begin{proof}
From Lemma~\ref{lem:p4}, we have $\rho_0 = 2^k - 1 - 2K$, and from Lemma~\ref{lem:D1972}, we have $(w_2-w_1)\rho_0 = n 2^k - (w_1+w_2)(2^k-1)$. 
The expression for $K$ then follows immediately.
Next, Lemma~\ref{lem:p4} gives
$$\left\{
\begin{array}{l}
\rho_1+\rho_2=2(\mu-\lambda-1), \\
\rho_1 \cdot \rho_2=2\lambda+2\mu-4K+1,
\end{array}
\right.$$
while Lemma~\ref{lem:D1972} yields $(w_2-w_1)\rho_1 = w_1+w_2$ and $(w_2-w_1)\rho_2 = w_1+w_2-2^k$.
Consequently, 
$$\left\{
\begin{array}{l}
(w_2-w_1)(\rho_1+\rho_2)=2(w_1+w_2)-2^k, \\
(w_2-w_1)^2 \rho_1 \rho_2 = (w_1+w_2)^2- (w_1+w_2) 2^k.
\end{array}
\right.$$
Substituting the identities for $\rho_1+\rho_2$ and $\rho_1\rho_2$ into the above system and simplifying yields
$$\left\{
\begin{array}{l}
2(\mu-\lambda-1)=\frac{2(w_1+w_2)-2^k}{w_2-w_1}, \\
2\lambda+2\mu-4K+1= \frac{(w_1+w_2)^2-(w_1+w_2)2^k}{(w_2-w_1)^2}
\end{array}
\right.
\Rightarrow
\left\{
\begin{array}{l}
2(\mu-\lambda)=\frac{4w_2-2^k}{w_2-w_1}, \\
2(\mu+\lambda)= \frac{4w_1w_2-(w_1+w_2)2^k}{(w_2-w_1)^2}+4K
\end{array}
\right.$$
Solving this system for $\mu$ and $\lambda$ gives the desired conclusion.
\end{proof}

\begin{remark}
Multiplying Equation (\ref{eq:nag}) by $J_v$ and comparing coefficients gives
\begin{equation*}
  \rho_0^2-\rho_0(\rho_1+\rho_2)-v(v-1+\rho_1\rho_2)+\rho_1\rho_2=0.
\end{equation*}
Then we have
\begin{equation*}
v=2^k=4w_1w_2/(n^2+n+4w_1w_2-2n(w_1+w_2))
\end{equation*}
if $n^2+n+4w_1w_2>2n(w_1+w_2)$.
That is, for some binary two-weight projective codes, the dimension $k$ can be determined from the length $n$ and the weights $w_1,w_2$.
\end{remark}

With the explicit parameters of the associated strongly regular graph established in Lemma \ref{lem:values}, 
we are now in a position to state the main result of this subsection.

\begin{theorem}\label{thm:2weight}
Let $C$ be a binary two-weight $[n,k]$ projective code with weight enumerator
$1+A_{w_1}z^{w_1}+A_{w_2}z^{w_2}$.
Then the asymmetric Hamming bidistance distribution of $C\backslash \{\mathbf{0}\}$ is given in Table \ref{tab:2weight}, 
where $\lambda$ and $\mu$ are the parameters of the associated strongly regular graph as determined in Lemma \ref{lem:values}.
\end{theorem}

\begin{proof}
Let $\mathbf{x},\mathbf{y}\in C\backslash\{\mathbf{0}\}$ be two distinct codewords.
Since $C$ is linear, their sum $\mathbf{z}=\mathbf{x}+\mathbf{y}$ is also a non-zero codeword, i.e.,
$\mathbf{z}\in C\backslash\{\mathbf{0}\}$. As ${\cal C}$ is a two-weight code, the weights of $\mathbf{x},\mathbf{y}$ and $\mathbf{z}$ can each take only the values $w_1$ or $w_2$. 
This leads to $2^3=8$ possible cases for the triple $(wt({\bf x}),wt({\bf y}),wt({\bf z}))$. 
For each case, the asymmetric Hamming bidistance $(d_{10}({\bf x},{\bf y}),d_{01}({\bf x},{\bf y}))$ is uniquely determined by the linear system in (\ref{eq:system}). 
Let the total number of ordered pairs $(\mathbf{x},\mathbf{y})$ corresponding to each of these eight cases be denoted by $f_1,f_2,\ldots,f_8$.

Let $G$ be the strongly regular graph associated with $C$, with parameters $(v,K,\lambda,\mu)$ as established in Lemma \ref{lem:values}. 
Recall that $v=|C|=2^k$ and importantly, the valency $K$ equals the number of codewords of weight $w_1$, i.e., $K=A_{w_1}$. 
Consequently, the number of codewords of weight $w_2$ is $A_{w_2}=v-1-A_{w_1}$.

We now proceed with the case-by-case analysis.
\begin{enumerate}
\item $wt(\mathbf{x})=wt(\mathbf{y})=wt(\mathbf{z})=w_1$.

Solving (\ref{eq:system}) yields $(d_{10}(\mathbf{x},\mathbf{y}),d_{01}(\mathbf{x},\mathbf{y}))=(\frac{w_1}{2},\frac{w_1}{2})$. 
The condition $wt(\mathbf{z})=w_1$ implies that $\mathbf{x}$ and $\mathbf{y}$ are adjacent. 
Consequently, the zero codeword $\bf 0$, together with $\mathbf{x}$ and $\mathbf{y}$ forms a triangle in $G$. 
Now ${\bf x}$ can be chosen arbitrarily among the $A_{w_1}$ codewords of weight $w_1$. 
Once ${\bf x}$ is fixed, ${\bf y}$ must be one of the $\lambda$ common neighbors of ${\bf 0}$ and ${\bf x}$ in $G$, as per the definition of a strongly regular graph. 
Hence, the number of ordered pairs $({\bf x},{\bf y})$ satisfying this case is $f_1=\lambda\cdot A_{w_1}$.

\item[2)] $wt(\mathbf{x})=wt(\mathbf{y})=w_1$ and $wt(\mathbf{z})=w_2$.
  
Here we have $(d_{10}(\mathbf{x},\mathbf{y}),d_{01}(\mathbf{x},\mathbf{y}))=(\frac{w_2}{2},\frac{w_2}{2})$. 
The condition $wt(\mathbf{z})=w_2$ means that $\mathbf{x}$ and $\mathbf{y}$ are non-adjacent in $G$. 
The total number of ordered pairs of distinct $w_1$-weight codewords is $A_{w_1}(A_{w_1}-1)$. 
Since these pairs are partitioned into those where their sum has weight $w_1$ (Case 1) and those where their sum has weight $w_2$ (Case 2), we have $f_1+f_2=A_{w_1}(A_{w_1}-1)$. 
Consequently, $f_2=A_{w_1}(A_{w_1}-\lambda-1).$
 
\item[3)] $wt(\mathbf{x})=wt(\mathbf{y})=wt(\mathbf{z})=w_2$. 
  
Now we have $(d_{10}(\mathbf{x},\mathbf{y}),d_{01}(\mathbf{x},\mathbf{y}))=(\frac{w_2}{2},\frac{w_2}{2})$. 
In this case, $\bf 0$, $\mathbf{x}$ and $\mathbf{y}$ form a triangle in the complement graph $\bar G$, which is also strongly regular. 
By an argument analogous to Case 1 applied to $\bar G$, we obtain $f_3=\bar{\lambda}\cdot A_{w_2}$, where $\bar{\lambda}=v-2A_{w_1}+\mu-2$ is the corresponding parameter of $\bar G$.

\item[4)] $wt(\mathbf{x})=wt(\mathbf{y})=w_2$ and $wt(\mathbf{z})=w_1$.

Now $(d_{10}(\mathbf{x},\mathbf{y}),d_{01}(\mathbf{x},\mathbf{y}))=(\frac{w_1}{2},\frac{w_1}{2})$.
Since $f_3+f_4=A_{w_2}(A_{w_2}-1)$, we have $f_4=A_{w_2}(A_{w_2}-\bar{\lambda}-1).$
 \end{enumerate} 

The remaining cases involve codewords of mixed weights. 
Due to the linearity and symmetry of the code, the frequencies for these cases are directly related to those already computed.
\begin{enumerate} 
\item[5)] If $wt(\mathbf{x})=wt(\mathbf{z})=w_1$ and $wt(\mathbf{y})=w_2$, then $(d_{10}(\mathbf{x},\mathbf{y}),d_{01}(\mathbf{x},\mathbf{y}))=(w_1-\frac{w_2}{2},\frac{w_2}{2})$, and $f_5=f_2.$
\item[6)] If $wt(\mathbf{x})=w_1$ and $wt(\mathbf{y})=wt(\mathbf{z})=w_2$, then $(d_{10}(\mathbf{x},\mathbf{y}),d_{01}(\mathbf{x},\mathbf{y}))=(\frac{w_1}{2},w_2-\frac{w_1}{2})$, and $f_6=f_4.$
\item[7)] If $wt(\mathbf{x})=w_2$ and $wt(\mathbf{y})=wt(\mathbf{z})=w_1$, then $(d_{10}(\mathbf{x},\mathbf{y}),d_{01}(\mathbf{x},\mathbf{y}))=(\frac{w_2}{2},w_1-\frac{w_2}{2})$, and $f_7=f_5=f_2.$
\item[8)] If $wt(\mathbf{x})=wt(\mathbf{z})=w_2$ and $wt(\mathbf{y})=w_1$, then $(d_{10}(\mathbf{x},\mathbf{y}),d_{01}(\mathbf{x},\mathbf{y}))=(w_2-\frac{w_1}{2},\frac{w_1}{2})$, and $f_8=f_6=f_4.$
\end{enumerate}
By aggregating the results from all eight cases and translating the frequencies $f_i(i=1,2\ldots,8)$ into the corresponding bidistance pairs, 
we obtain the complete asymmetric Hamming bidistance distribution for $C\setminus\{\bf 0\}$, as presented in Table \ref{tab:2weight}. 
This completes the proof.
\end{proof}

\begin{table*}[!htbp]
\centering
\caption{The distribution of asymmetric Hamming bidistance of $C\backslash \{\mathbf{0}\}$}\label{tab:2weight}
\vspace{-0.1cm}
\resizebox{0.66\linewidth}{!}{
\renewcommand{\arraystretch}{1.4}
\begin{tabular}{c|c}
\hline
Asymmetric Hamming bidistance & Frequency                             \\ \hline
$(\frac{w_1}{2},\frac{w_1}{2})$            & $A_{w_1}A_{w_2}+\lambda A_{w_1}-\mu A_{w_2}$ \\ \hline
$(\frac{w_2}{2},\frac{w_2}{2})$            & $A_{w_2}(A_{w_2}-A_{w_1}+\mu-1)+A_{w_1}(A_{w_1}-\lambda-1)$  \\ \hline
$(w_1-\frac{w_2}{2},\frac{w_2}{2})$        & $A_{w_1}(A_{w_1}-\lambda-1)$          \\ \hline
$(\frac{w_1}{2},w_2-\frac{w_1}{2})$        & $A_{w_2}(A_{w_1}-\mu)$                \\ \hline
$(\frac{w_2}{2},w_1-\frac{w_2}{2})$        & $A_{w_1}(A_{w_1}-\lambda-1)$          \\ \hline
$(w_2-\frac{w_1}{2},\frac{w_1}{2})$        & $A_{w_2}(A_{w_1}-\mu)$                \\ \hline
\end{tabular}}
\end{table*}

\begin{remark}
$(1)$ In Theorem \ref{thm:2weight}, as previously described, for $C\setminus\{\mathbf 0\}$, the omitted diagonal entry is $A(0,0)=|C\setminus\{\mathbf 0\}|=2^k-1.$
We omit it because it counts only the trivial ordered pairs $(\mathbf x,\mathbf x)$.
The same convention applies to Tables \ref{tab:3weight} and \ref{tab:SBIBD} below.

$(2)$ Table \ref{tab:2weight} (including Table \ref{tab:3weight} in the sequel) is presented in a casewise formal form.
If a listed frequency is zero, the corresponding bidistance pair is to be omitted. 
If two or more displayed bidistance pairs coincide for a particular choice of parameters, 
their frequencies must be added before writing the resulting AHB distribution.
\end{remark}

We now present an example to illustrate Theorem \ref{thm:2weight}.

\begin{exmpl}\label{exp:6}
Let $D=\{x\in \F_{64}^*:Tr_{2}^{8}(x^{9})=0\},$ where $Tr_{2}^{8}(\cdot)$ denotes the trace function from $\F_{8}$ to $\F_2$. 
Noting that $(x^9)^7=x^{63}=1$ for all $x\in \F_{64}^*$, we have $x^9 \in \F_8$, thereby rendering the definition of $D$ valid.
Using MAGMA, we obtain $|D|=27$.
Write $D=\{d_1,d_2,\ldots,d_{27}\}$ and define 
$$C_D=\{\left(Tr_{2}^{64}(\beta d_1),Tr_{2}^{64}(\beta d_2),\ldots,Tr_{2}^{64}(\beta d_{27})\right):\beta \in \F_{64}\}.$$
Then $C_D$ is a binary two-weight $[27,6]$ projective code with weight enumerator
$1+36z^{12}+27z^{16}$.
By Lemma \ref{lem:D1972} and Lemma \ref{lem:values}, the strongly regular graph $G$ associated with $C_D$ has parameters $(64,36,20,20)$.
Applying Theorem \ref{thm:2weight}, the asymmetric Hamming bidistance distribution of $C_D\backslash \{\mathbf{0}\}$ is obtained as shown in Table \ref{tab:exp1}.
\begin{table}[!h]
\centering
\caption{Asymmetric Hamming bidistance distribution of $C_D\backslash \{\mathbf{0}\}$}\label{tab:exp1}
\resizebox{0.4\linewidth}{!}{
\begin{tabular}{cc|cc}
\hline
AHB & Frequency  & AHB & Frequency \\ \hline
$(6,6)$                    & 1152       & $(6,10)$                   & 432       \\ 
$(8,8)$                    & 810        & $(8,4)$                    & 540       \\ 
$(4,8)$                    & 540        & $(10,6)$                   & 432       \\ \hline
\end{tabular}}
\end{table}
\end{exmpl}

\subsection{Three-weight Codes via 3-class Association Schemes}\label{sec:4-2}
In this subsection, we employ 3-class association schemes to determine the asymmetric Hamming bidistance and its distribution for some three-weight projective codes (with the zero codeword removed). 
We begin by recalling the necessary definitions and relevant results.

A symmetric \emph{association scheme} with $s$ classes consists of a finite set $X$ of $v$ points together with $s+1$ relations $R_0,R_1,\ldots,R_s$ on $X$ satisfying the following conditions:
\begin{enumerate}
\item $R_0,R_1,\ldots,R_s$ form a partition of $X\times X$ and $R_0$ is the diagonal relation, i.e., $R_0=\{(x,x) \mid x\in X\}$;
\item $(x,y)\in R_k$ if and only if $(y,x)\in R_k$;
\item if $(x,y)\in R_k,$ the number of $z \in X$ such that $(x,z)\in R_i,$ and $(z,y)\in R_j,$ is a constant $p^k_{ij}$ depending only on $i,j,k$ and not on the particular choice of $x$ and $y$.
\end{enumerate}
The constants $p^k_{ij}$ are called the \emph{intersection numbers} of the scheme.
In particular, $v_i= p^0_{ii},i\in \{0,1,\ldots,s\},$ is called the \emph{valence} of the relation $R_i$. 
For each $R_i$, we define the  \emph{adjacency matrix} $D_i$ as  the $v\times v$ matrix whose rows and columns are indexed by the elements of $X$, with entries given by
\begin{equation*}
D_{i}(x,y)=\left\{
\begin{array}{cl}
  1, & \text{ if } (x,y)\in R_i,\\
  0, & \text{ otherwise. } 
\end{array}
\right.
\end{equation*}
From the definition, we have $D_0=I_v$, each $D_i$ is symmetric, and the following relations hold:
\begin{equation}\label{association_shm}
\sum_{i=0}^{s}D_i=J_v,~~~D_i D_j=\sum_{k=0}^{s} p^k_{ij} D_k,~~i,j=0,1,\ldots,s.
\end{equation}

Let $C$ be a three-weight $[n,k]$ projective code over $\F_2$ with nonzero weights $w_1<w_2<w_3$ (and $w_0=0$ for the zero codeword). 
Based on the Hamming distances among codewords, we define relations 
$$R_i=\{(\mathbf{c}_1,\mathbf{c}_2)\in C\times C: d_H(\mathbf{c}_1,\mathbf{c}_2)=w_i\},~i=0,1,2,3.$$
The pair ($C$,$\{R_i\}_{i=0}^3)$ is called the \emph{distance scheme} of $C$.
Let $C^{\bot}$ be the dual code of $C$.  
This structure is closely related to the dual code $C^{\bot}$. 
Let the \emph{distribution matrix} of $C^{\bot}$ be the matrix whose rows record the weight distributions of all cosets in $\F_2^n$ of $C^{\bot}$. 
The following lemma, due to Delsarte \cite{Delsarte1973}, characterizes when the distance scheme forms a 3-class association scheme.

\begin{lemma}\cite{Delsarte1973}
\label{lem:conditions}
The distance scheme of a three-weight projective code $C$ is a $3$-class association scheme if and only if the distribution matrix of the dual code $C^{\bot}$ contains four distinct rows.
\end{lemma}

\begin{table*}[!htbp]
\centering
\caption{The distribution of asymmetric Hamming bidistance of $C\backslash \{\mathbf{0}\}$}\label{tab:3weight}
\vspace{-0.1cm}
\resizebox{0.8\linewidth}{!}{
\begin{threeparttable}
\renewcommand{\arraystretch}{1.5}
\begin{tabular}{c|c||c|c}
\hline
Asymmetric Hamming bidistance & Frequency                                                               & Asymmetric Hamming bidistance              & Frequency         \\ \hline
$(\frac{w_1}{2},\frac{w_1}{2})$            & $v_1 p_{11}^{1}+v_2 p_{12}^2+v_3 p_{13}^3$                 & $(w_1-\frac{w_2}{2},\frac{w_2}{2})$        & $v_1 p_{12}^{1}$  \\ \hline
$(\frac{w_2}{2},\frac{w_2}{2})$            & $v_1 p_{12}^{1}+v_2 p_{22}^2+v_3 p_{23}^3$                 & $(w_1-\frac{w_3}{2},\frac{w_3}{2})$        & $v_1 p_{13}^{1}$  \\ \hline
$(\frac{w_3}{2},\frac{w_3}{2})$            & $v_1 p_{13}^{1}+v_2 p_{23}^2+v_3 p_{33}^3$                 & $(\frac{w_1}{2},w_2-\frac{w_1}{2})$        & $v_2 p_{12}^{2}$  \\ \hline
$(\frac{w_1+w_3-w_2}{2},\frac{w_2+w_3-w_1}{2})$        & $v_1 v_2-v_1 p_{12}^{1}-v_2 p_{12}^{2}$        & $(\frac{w_1}{2},w_3-\frac{w_1}{2})$        & $v_3 p_{13}^{3}$  \\ \hline
$(\frac{w_1+w_2-w_3}{2},\frac{w_2+w_3-w_1}{2})$        & $v_1 v_2-v_1 p_{12}^{1}-v_2 p_{12}^{2}$        & $(w_2-\frac{w_3}{2},\frac{w_3}{2})$        & $v_2 p_{23}^{2}$  \\ \hline
$(\frac{w_1+w_2-w_3}{2},\frac{w_1+w_3-w_2}{2})$        & $v_1 v_2-v_1 p_{12}^{1}-v_2 p_{12}^{2}$        & $(\frac{w_2}{2},w_3-\frac{w_2}{2})$        & $v_3 p_{23}^{3}$  \\ \hline
\end{tabular}
\begin{tablenotes}
\footnotesize
\item There are another 9 cases, which exhibit a symmetric distribution with respect to the asymmetric Hamming bidistance along with the latter 9 cases in the table and have equal frequencies, thus, they are omitted here.
For instance, the bidistance $(w_1-\frac{w_2}{2},\frac{w_2}{2})$ has a symmetric counterpart $(\frac{w_2}{2},w_1-\frac{w_2}{2})$ with the same frequency $v_1 p_{12}^{1}$.
\end{tablenotes}
\end{threeparttable}}
\end{table*}

\begin{table*}[!htbp]
\centering
\caption{The parameters of the 27 cases in the proof of Theorem \ref{thm:3weight}}
\label{tab:proof}
\renewcommand{\arraystretch}{1.4}
\resizebox{0.85\linewidth}{!}{
\begin{tabular}{|c|c|c|c|c|c|}
\hline
Case & $wt(\mathbf{x})$ & $wt(\mathbf{y})$ & $wt(\mathbf{z})$ & $(d_{10}(\mathbf{x},\mathbf{y}),d_{01}(\mathbf{x},\mathbf{y}))$ & Frequency     \\ \hline
$(1)$ & $w_1$            & $w_1$            & $w_1$            & $(\frac{w_1}{2},\frac{w_1}{2})$                                 & $f_1=v_1 \cdot p_{11}^1$ \\ \hline
$(2)$ & $w_1$            & $w_1$            & $w_2$            & $(\frac{w_2}{2},\frac{w_2}{2})$                                 & $f_2=v_1 \cdot p_{12}^1$ \\ \hline
$(3)$ & $w_1$            & $w_1$            & $w_3$            & $(\frac{w_3}{2},\frac{w_3}{2})$                                 & $f_3=v_1 \cdot p_{13}^1$ \\ \hline
$(4)$ & $w_2$            & $w_2$            & $w_1$            & $(\frac{w_1}{2},\frac{w_1}{2})$                                 & $f_4=v_2\cdot p_{12}^2$  \\ \hline
$(5)$ & $w_2$            & $w_2$            & $w_2$            & $(\frac{w_2}{2},\frac{w_2}{2})$                                 & $f_5=v_2 \cdot p_{22}^2$ \\ \hline
$(6)$ & $w_2$            & $w_2$            & $w_3$            & $(\frac{w_3}{2},\frac{w_3}{2})$                                 & $f_6=v_2 \cdot p_{23}^2$ \\ \hline
$(7)$ & $w_3$            & $w_3$            & $w_1$            & $(\frac{w_1}{2},\frac{w_1}{2})$                                 & $f_7=v_3 \cdot p_{13}^3$ \\ \hline
$(8)$ & $w_3$            & $w_3$            & $w_2$            & $(\frac{w_2}{2},\frac{w_2}{2})$                                 & $f_8=v_3 \cdot p_{23}^3$ \\ \hline
$(9)$ & $w_3$            & $w_3$            & $w_3$            & $(\frac{w_3}{2},\frac{w_3}{2})$                                 & $f_9=v_3 \cdot p_{33}^3$ \\ \hline
$(10)$ & $w_1$            & $w_2$            & $w_1$            & $(w_1-\frac{w_2}{2},\frac{w_2}{2})$                            & $f_{10}=f_2$ \\ \hline
$(11)$ & $w_1$            & $w_3$            & $w_1$            & $(w_1-\frac{w_3}{2},\frac{w_3}{2})$                            & $f_{11}=f_3$ \\ \hline
$(12)$ & $w_1$            & $w_2$            & $w_2$            & $(\frac{w_1}{2},w_2-\frac{w_1}{2})$                            & $f_{12}=f_4$ \\ \hline
$(13)$ & $w_2$            & $w_3$            & $w_2$            & $(w_2-\frac{w_3}{2},\frac{w_3}{2})$                            & $f_{13}=f_6$ \\ \hline
$(14)$ & $w_1$            & $w_3$            & $w_3$            & $(\frac{w_1}{2},w_3-\frac{w_1}{2})$                            & $f_{14}=f_7$ \\ \hline
$(15)$ & $w_2$            & $w_3$            & $w_3$            & $(\frac{w_2}{2},w_3-\frac{w_2}{2})$                            & $f_{15}=f_8$ \\ \hline
$(16)$ & $w_2$            & $w_1$            & $w_1$            & $(\frac{w_2}{2},w_1-\frac{w_2}{2})$                            & $f_{16}=f_2$ \\ \hline
$(17)$ & $w_3$            & $w_1$            & $w_1$            & $(\frac{w_3}{2},w_1-\frac{w_3}{2})$                            & $f_{17}=f_3$ \\ \hline
$(18)$ & $w_2$            & $w_1$            & $w_2$            & $(w_2-\frac{w_1}{2},\frac{w_1}{2})$                            & $f_{18}=f_4$ \\ \hline
$(19)$ & $w_3$            & $w_2$            & $w_2$            & $(\frac{w_3}{2},w_2-\frac{w_3}{2})$                            & $f_{19}=f_6$ \\ \hline
$(20)$ & $w_3$            & $w_1$            & $w_3$            & $(w_3-\frac{w_1}{2},\frac{w_1}{2})$                            & $f_{20}=f_7$ \\ \hline
$(21)$ & $w_3$            & $w_2$            & $w_3$            & $(w_3-\frac{w_2}{2},\frac{w_2}{2})$                            & $f_{21}=f_8$ \\ \hline

$(22)$ & $w_1$            & $w_2$            & $w_3$            & $(\frac{w_1+w_3-w_2}{2},\frac{w_2+w_3-w_1}{2})$                & $f_{22}=v_1\cdot v_2-f_2-f_4$ \\ \hline
$(23)$ & $w_1$            & $w_3$            & $w_2$            & $(\frac{w_1+w_2-w_3}{2},\frac{w_2+w_3-w_1}{2})$                & $f_{23}=f_{22}$ \\ \hline
$(24)$ & $w_2$            & $w_3$            & $w_1$            & $(\frac{w_1+w_2-w_3}{2},\frac{w_1+w_3-w_2}{2})$                & $f_{24}=f_{22}$ \\ \hline

$(25)$ & $w_2$            & $w_1$            & $w_3$            & $(\frac{w_2+w_3-w_1}{2},\frac{w_1+w_3-w_2}{2})$                & $f_{25}=f_{22}$ \\ \hline
$(26)$ & $w_3$            & $w_1$            & $w_2$            & $(\frac{w_2+w_3-w_1}{2},\frac{w_1+w_2-w_3}{2})$                & $f_{26}=f_{22}$ \\ \hline
$(27)$ & $w_3$            & $w_2$            & $w_1$            & $(\frac{w_1+w_3-w_2}{2},\frac{w_1+w_2-w_3}{2})$                & $f_{27}=f_{22}$ \\ \hline
\end{tabular}}
\end{table*}

We are now in a position to present the main result of this subsection.

\begin{theorem}\label{thm:3weight}
Let $C$ be a binary three-weight $[n,k]$ projective code with weight enumerator
$1+A_{w_1}z^{w_1}+A_{w_2}z^{w_2}+A_{w_3}z^{w_3}$ and suppose $C$ satisfies the condition in Lemma \ref{lem:conditions}. 
The asymmetric Hamming bidistance distribution of $C\backslash \{\mathbf{0}\}$ is given in Table \ref{tab:3weight},
where $v_i=A_{w_i}$ and $p_{ij}^k$ are the intersection numbers of the 3-class association scheme corresponding to $C$.
\end{theorem}
\begin{proof}
Following the same approach as in the proof of Theorem \ref{thm:2weight}, we analyze all possible triples $(wt({\bf x}),wt({\bf y}),wt({\bf z}))$ with ${\bf z}={\bf x}+{\bf y}$ (${\bf z}\ne 0$). 
Since each of the three weights can take any of the three nonzero values $w_1,w_2,w_3$, there are $3^3=27$ distinct cases. 
Let the frequencies of ordered pairs $(\mathbf{x},\mathbf{y})$ corresponding to these cases be denoted by $f_1,f_2,\ldots,f_{27}$. 

By Lemma \ref{lem:conditions}, the code $C$ gives rise to a 3-class association scheme with relations
\begin{equation*}
R_i=\{(\mathbf{c}_1,\mathbf{c}_2)\in C\times C: d_H(\mathbf{c}_1,\mathbf{c}_2)=w_i\},~~i=0,1,2,3,
\end{equation*}
where $w_0=0$. From the definition of an association scheme, the valences are $v_i=A_{w_i}$ for $i=1,2,3$, and the intersection numbers $p_{ij}^k$ satisfy
\begin{center}
$p_{ij}^k=|\{{\bf y}\in C:({\bf x},{\bf y})\in R_i~{\rm and}~({\bf y},{\bf z})\in R_j\}|$ for any $({\bf x},{\bf z})\in R_k$.	
\end{center}
We now determine the frequencies $f_1$ through $f_{27}$ using these parameters.

\paragraph{Cases 1-9}
For each ordered triple of weights $(wt({\bf x}),wt({\bf y}),wt({\bf z}))$, the asymmetric bidistance $(d_{10},d_{01})$ is uniquely determined by (\ref{eq:system}). 
The frequency of such a triple can be expressed in terms of the valences and intersection numbers. 
For instance, in Case 1 where $wt(\mathbf{x})=wt(\mathbf{y})=wt(\mathbf{z})=w_1$, we have $(\mathbf{0},\mathbf{x})\in R_1$ and $(\mathbf{0},\mathbf{y})\in R_1$. 
The condition $wt({\bf z})=w_1$ implies $({\bf x},{\bf y})\in R_1$. By the definition of intersection numbers, 
the number of ${\bf y}$ satisfying these relations for a fixed ${\bf x}$ is $p_{11}^1$. 
Since ${\bf x}$ can be any of the $v_1$ codewords of weight $w_1$, we obtain  $f_1=v_1 \cdot p_{11}^1$. 
The remaining cases from Case 2 to Case 9 are obtained similarly by applying the appropriate intersection numbers $p_{ij}^k$ corresponding to the relations among ${\bf 0},{\bf x},{\bf y}$ and ${\bf z}$. 
The explicit expressions are summarized in Table \ref{tab:proof}.

\paragraph{Cases 10-15}
By symmetry of the roles of ${\bf x},{\bf y},{\bf z}$, the frequencies for cases where the weight triples are permutations of each other coincide. 
Specifically, Cases 10-15 correspond to permutations of the triples already encountered in Cases 2, 3, 4, 6, 7, and 8, respectively. 
Hence we have $f_{10}=f_2,f_{11}=f_3, f_{12}=f_4,f_{13}=f_6,f_{14}=f_7,f_{15}=f_8$.

\paragraph{Cases 16-21}
Similarly, Cases 16-21 are permutations of Cases 2, 3, 4, 6, 7, and 8 from the perspective of 
${\bf z}$, yielding $f_{16}=f_2,f_{17}=f_3, f_{18}=f_4,f_{19}=f_6,f_{20}=f_7,f_{21}=f_8$.

\paragraph{Cases 22-27}
Consider Case 22, where the weight triple is $(w_1,w_2,w_3)$. The total number of ordered pairs $({\bf x},{\bf y})$ with $wt({\bf x})=w_1$ and $wt({\bf y})=w_2$ is $v_1v_2$. 
These pairs are distributed among Cases 10, 12, and 22 according to the weight of ${\bf z}={\bf x}+{\bf y}$. 
Therefore, $f_{22}=v_1\cdot v_2-f_{10}-f_{12}=v_1\cdot v_2-f_{2}-f_{4}$. By symmetry, the remaining cases in this group satisfy $f_{23}=f_{24}=f_{25}=f_{26}=f_{27}=f_{22}$. 

Collecting all frequencies and mapping them to the corresponding asymmetric Hamming bidistance pairs yields the distribution shown in Table \ref{tab:3weight}. 
This completes the proof.
\end{proof}

The following example is given to further clarify Theorem \ref{thm:3weight}.

\begin{exmpl}\label{exp:7}
The binary Golay code $C$ is a linear code with parameters $[23,12]$, 
and its distribution matrix  has four distinct rows.
The dual code $C^{\bot}$ is a $[23,11]$ three-weight code with weight enumerator
$$1+506z^{8}+1288z^{12}+253z^{16}.$$
Consequently, the distance scheme of $C^{\bot}$ forms a $3$-class association scheme.
From (\ref{association_shm}), we obtain the valences and intersection numbers:
$$\begin{array}{lll}
   v_1=506, & v_2=1288, & v_3=253,\\
   p_{11}^1=210, & p_{12}^2=330, & p_{13}^3=140, \\
   p_{12}^1=280, & p_{22}^2=792, & p_{23}^3=112,\\
   p_{13}^1=15, & p_{23}^2=165, & p_{33}^3=0.
\end{array}$$
Applying Theorem \ref{thm:3weight} yields the asymmetric Hamming bidistance distribution of $C^{\bot} \backslash \{\mathbf{0}\}$, which is presented in Table \ref{tab:exp2}.
\begin{table}[!hbpt]
\centering
\caption{The asymmetric Hamming bidistance  distribution of $C^{\bot} \backslash \{\mathbf{0}\}$}\label{tab:exp2}
\resizebox{0.4\linewidth}{!}{
\renewcommand{\arraystretch}{1.2}
\begin{tabular}{cc}
\hline
Asymmetric Hamming bidistance & Frequency of each bidistance  \\
\hline
$(4,4)$                    & 566720  \\
$(6,6)$                    & 1190112  \\
$(8,8)$                    & 220110  \\
$(0,8),(8,0)$              & 7590    \\
$(2,6),(6,2)$              & 226688 \\
$(2,10),(10,2)$            & 85008  \\
$(4,8),(8,4)$              & 637560  \\
$(4,12),(12,4)$            & 35420  \\
$(6,10),(10,6)$            & 113344    \\
\hline
\end{tabular}}
\end{table}
\end{exmpl}

The following example illustrates how the AHB distributions obtained in Examples~\ref{exp:6} and~\ref{exp:7} can be used to evaluate the bound in Theorem~\ref{Refined-bound}
and to compare it numerically with the two discrepancy-based bounds in Lemma~\ref{lem:bound}.

\begin{exmpl}\label{ex:bounds-few-weight}
Denote by $\mathcal D_{1}=C_D\setminus\{\mathbf 0\}$ the two-weight code in Example~\ref{exp:6}, 
and by $\mathcal D_{2}=C^\perp\setminus\{\mathbf 0\}$ the three-weight code in Example~\ref{exp:7}. 
Let $B_{\mathrm{AHB}}$ be the bound in Theorem~\ref{Refined-bound}, 
and let $B_{\delta_\gamma}$ and $B_{\hat{\delta}_\gamma}$ be the first and second discrepancy-based bounds in Lemma~\ref{lem:bound}, respectively.
For $\mathcal D_1$, Table~\ref{tab:exp1} and its weight enumerator give
$$\delta_\gamma(\mathcal D_1)=4\gamma+8,~~\hat{\delta}_\gamma(\mathcal D_1)=22-10\gamma.$$
Similarly, for $\mathcal D_2$, Table~\ref{tab:exp2} and its weight enumerator give
$$\delta_\gamma(\mathcal D_2)=8,~~\hat{\delta}_\gamma(\mathcal D_2)=20-12\gamma.$$

Following the approach in Example \ref{exp:5}, it is possible to express the three bounds as piecewise polynomials of $ p $ and $ q $. 
However, due to larger code sizes and an increased number of bidistance classes, the polynomial expressions become much longer. 
Additionally, the thresholds in Theorem \ref{Refined-bound} may vary with $ \gamma $, necessitating separate treatment of different polynomial regimes.

Since the comparison mechanism is similar to that in Example \ref{exp:5}, where bounds reduce to polynomials based on $ p $ and $ q $, 
we will not repeat the full symbolic expansion. 
Instead, we provide numerical comparisons to demonstrate the performance of the AHB-based bound for these structured codes. 
Substituting the AHB distributions in Tables~\ref{tab:exp1} and~\ref{tab:exp2} into Theorem \ref{Refined-bound}, and applying the two bounds in Lemma~\ref{lem:bound}, 
we obtain, for the representative channel condition $(p,q)=(0.01,0.1)$,
\[\begin{array}{c|c|c|c}
\text{Code}	    &   B_{\mathrm{AHB}}        &   B_{\delta_\gamma}       &   B_{\hat{\delta}_\gamma}	\\ \hline
\mathcal D_1 	&	  1.368\times10^{-4}	  &	  2.9\times10^{-3}	      &	  6.856\times10^{-3}	    \\
\mathcal D_2	&	  0.091		              &	  0.1301		          &   0.1564		            \\
\end{array}	\]
Thus, under this channel condition, the AHB-based bound is strictly sharper than both discrepancy-based bounds for both the two-weight code $\mathcal D_1$ and the three-weight code $\mathcal D_2$.

Moreover, for the larger transition probabilities $(p,q)=(0.1,0.2)$, the AHB-based bound for $\mathcal D_1$ remains strictly smaller than both discrepancy-based bounds. 
More precisely, both discrepancy-based bounds evaluate to approximately $0.2124$, whereas the AHB-based bound is approximately $0.1221$.
\end{exmpl}

\section{Asymmetric Hamming Bidistance Distributions of SBIBD-Derived Codes}\label{sec:6}

In a recent work \cite{ZWC2025}, we constructed several families of binary \emph{combinatorial neural (CN) codes} derived from symmetric balanced incomplete block designs (SBIBDs). 
These codes were shown to be optimal with respect to the improved Plotkin bound on the discrepancy measure $\delta_\gamma$. In this section, we further determine their full asymmetric Hamming bidistance distributions.

We begin by recalling the construction of these codes.

A \emph{balanced incomplete block design (BIBD)} with parameters $(v,k,\lambda)$ is a pair $(G,{\cal B})$,
where $G$ is a finite set of $v$ points and ${\cal B}$ is a collection of $k$-element subsets (called \emph{blocks}) of $G$ such that every pair of distinct points is contained in exactly $\lambda$ blocks.
In a $(v,k,\lambda)$ BIBD, the number of blocks containing a point $g\in G$ is a constant $\lambda(v-1)/(k-1)$, which is called the \emph{replication number} of $g$.  
For a non-degenerate BIBD (i.e., $2\leq k<v$), Fisher's Inequality \cite{CDS1999} states that the number of blocks satisfies $v\leq b=|{\cal B}|=\lambda v(v-1)/(k(k-1))$. 
A BIBD is called \emph{symmetric}, denoted by SBIBD, if the number of blocks equals the number of points, i.e., $v=b$. 
Given a SBIBD $(G,{\cal B})$, its \emph{complement} is defined as $(G,\bar{\cal B})$, where $\bar{\cal B}=\{\bar{B}=G\setminus B:B\in{\cal B}\}$. 

An SBIBD has several useful properties, some of which are listed below.
\begin{lemma}\cite{CDS1999}\label{SBIBD-lem}
Let $(G,{\cal B})$ be a $(v,k,\lambda)$-SBIBD. Then we have
\begin{enumerate}
	\item [\rm 1)] the replication number of each point in $G$ is equal to $k;$
	\item [\rm 2)] $|B\cap B'|=\lambda$ and $|B\cap \bar{B'}|=k-\lambda$  for any distinct $B,B'\in{\cal B};$ and
	\item [\rm 3)] its complement is a $(v,v-k,v-2k+\lambda)$-SBIBD$.$
\end{enumerate}
\end{lemma}

The SBIBDs with known parameters could have different point sets $G=\{g_1,g_2,\ldots,g_v\}.$
However, they all can be transformed into isomorphic SBIBDs with the point set $[v]=\{1,2,\ldots,v\}$ by the natural one-to-one correspondence $g_i\mapsto i$ for $1\le i\le v$.
Recall the definition of the support set of a binary vector in Section \ref{sec:2}.
Then a binary code ${\cal C}$ can correspond naturally to a collection of subsets of $G$, referred to as the support set of ${\cal C}$ in $G$. 
Based on this correspondence, several families of binary nonlinear codes were constructed from SBIBDs in \cite{ZWC2025}.

For a nonlinear binary code, the notation $(n,M)$ means that the code has length $n$ and cardinality $M$.
\begin{construction}\cite[Construction 1]{ZWC2025}\label{con}
Let $(G,{\cal B})$ be a $(v,k,\lambda)$-SBIBD with $v\geq 2k$. Then
\begin{enumerate}
	\item [\rm 1)] a $(v,v)$ binary code ${\cal C}_1$ is obtained with the support set in $G$ as ${\cal A}_1={\cal B}$;
	\item [\rm 2)] a $(v,2v)$ binary code ${\cal C}_2$ is obtained with the support set in $G$ as ${\cal A}_2={\cal B}\cup \bar{\cal B}$;
	\item [\rm 3)] a $(v+1,2v)$ binary code ${\cal C}_3$ is obtained with the support set in $G\cup\{g_{v+1}\}$ as
	\[{\cal A}_3=\big\{B\cup\{g_{v+1}\}: B\in{\cal B}\big\}\cup \bar{\cal B},\]
    where $g_{v+1}$ is a point not belonging to $G$;
\item [\rm 4)] let $g$ be a fixed point in $G$, then a $(v-1,v)$ binary code ${\cal C}_4$ is obtained with support set in $G\setminus\{g\}$ as
	\[{\cal A}_4=\big\{B\setminus\{g\}: B\in{\cal B},g\in B\big\}\cup \big\{\bar{B}\setminus\{g\}: B\in{\cal B},g\notin B \big\}.\]
\end{enumerate}	
\end{construction}

\begin{remark}
In Construction \ref{con}, the code ${\cal C}_3$ can be viewed as an extension of ${\cal C}_2$ by appending one additional coordinate. 
Specifically, for the codewords of ${\cal C}_2$ whose support sets belong to ${\cal B}$, a ``$1$" is inserted at the extended coordinate position $v+1$, 
and a ``$0$" is inserted at that position for other codewords.
Similarly, ${\cal C}_4$ can be obtained from a subset of ${\cal C}_2$ by deleting one coordinate.
More precisely, for a fixed point $g=g_i\in G$, take all codewords in ${\cal C}_2$ whose $i$-th coordinate is ``$1$",
and then remove that coordinate to form the code ${\cal C}_4$.
\end{remark}

The following theorem completely determines the asymmetric Hamming bidistance distributions of these codes.

\begin{theorem}\label{thm:SBIBD}
For each $i\in\{1,2,3,4\}$, let $\mathcal C_i$ be the binary nonlinear code obtained from Construction \ref{con}. 
Apart from the diagonal entry $A_{\mathcal C_i}(0,0)=|\mathcal C_i|,$
the nonzero entries of the asymmetric Hamming bidistance distribution of $\mathcal C_i$ are given in Table \ref{tab:SBIBD}.
\begin{table}[!h]
\centering
\caption{The asymmetric Hamming bidistance distributions of the codes in Theorem \ref{thm:SBIBD} }
\label{tab:SBIBD}
\resizebox{0.9\linewidth}{!}{
\renewcommand{\arraystretch}{1.25}
\begin{tabular}{|c|c|c|c|c|c|}
\hline
Code                          & Asymmetric Hamming bidistance               & Frequency                 & Code                          & Asymmetric Hamming bidistance & Frequency             \\ \hline
\multirow{3}{*}{${\cal C}_1$} & \multirow{3}{*}{$(k-\lambda,k-\lambda)$} & \multirow{3}{*}{$v(v-1)$} & \multirow{3}{*}{${\cal C}_4$} & $(k-\lambda,k-\lambda)$    & $k(k-1)+(v-k)(v-k-1)$ \\ \cline{5-6} 
                              &                                          &                           &                               & $(\lambda,v-2k+\lambda)$   & $k(v-k)$              \\ \cline{5-6} 
                              &                                          &                           &                               & $(v-2k+\lambda,\lambda)$   & $k(v-k)$              \\ \hline
\multirow{5}{*}{${\cal C}_2$} & $(k-\lambda,k-\lambda)$                  & $2v(v-1)$                 & \multirow{5}{*}{${\cal C}_3$} & $(k-\lambda,k-\lambda)$    & $2v(v-1)$             \\ \cline{2-3} \cline{5-6} 
                              & $(k,v-k)$                                & $v$                       &                               & $(k+1,v-k)$                & $v$                   \\ \cline{2-3} \cline{5-6} 
                              & $(v-k,k)$                                & $v$                       &                               & $(v-k,k+1)$                & $v$                   \\ \cline{2-3} \cline{5-6} 
                              & $(\lambda,v-2k+\lambda)$                 & $v(v-1)$                  &                               & $(\lambda+1,v-2k+\lambda)$ & $v(v-1)$              \\ \cline{2-3} \cline{5-6} 
                              & $(v-2k+\lambda,\lambda)$                 & $v(v-1)$                  &                               & $(v-2k+\lambda,\lambda+1)$ & $v(v-1)$              \\ \hline
\end{tabular}}
\end{table}
\end{theorem}
\begin{proof}
Let ${\bf x},{\bf y}\in{\cal C}_i$ be two distinct codewords, and let $B_{\bf x},B_{\bf y}\subseteq X$ be their corresponding supports as defined in Construction \ref{con}, 
where $X=G$ for ${\cal C}_1,{\cal C}_2$ and $X=G\cup\{g_{v+1}\}$ or $G\setminus\{g\}$ for ${\cal C}_3,{\cal C}_4$, respectively. 
From the definition of directional distances, we have
\begin{align*}
	d_{10}({\bf x},{\bf y})=|B_{\bf x}\cap \bar{B}_{\bf y}|,~~
	d_{01}({\bf x},{\bf y})=|\bar{B}_{\bf x}\cap B_{\bf y}|,
\end{align*}
where $\bar{B}$ denotes the complement of $B$ in the underlying point set $X$.

By Construction \ref{con}, the support sets of codewords in each ${\cal C}_i$ are specific families of subsets derived from the blocks of a $(v,k,\lambda)$-SBIBD $(G,{\cal B})$ and their complements. 
The required intersection sizes are obtained by reducing each case to the block-intersection identities in Lemma~\ref{SBIBD-lem}, together with the addition or deletion of the distinguished point.

We consider the four constructions separately.

\noindent
\textbf{The code ${\cal C}_1$.} 
For any two distinct $B_{\bf x}$ and $B_{\bf y}$, we have $|B_{\bf x}\cap B_{\bf y}|=\lambda$ by Lemma \ref{SBIBD-lem}, 
and hence $|B_{\bf x}\cap \bar{B}_{\bf y}| = |\bar{B}_{\bf x}\cap B_{\bf y}| = k-\lambda.$
Therefore, the only nonzero asymmetric Hamming bidistance of ${\cal C}_1$ is $(k-\lambda,k-\lambda),$ whose frequency is $v(v-1)$.

\noindent
\textbf{The code ${\cal C}_2$.} 
We decompose ${\cal C}_2$ into two subfamilies ${\cal C}_2^{(1)}$ and ${\cal C}_2^{(2)}$, 
where the support set families of ${\cal C}_2^{(1)}$ and ${\cal C}_2^{(2)}$ are ${\cal B}$ and $\bar{\cal B}$, respectively.
Thus $|{\cal C}_2^{(1)}|=|{\cal C}_2^{(2)}|=v$.
For any distinct codewords $\mathbf{x},\mathbf{y} \in {\cal C}_2^{(1)}$ (or ${\cal C}_2^{(2)}$), 
we have $d_{10}(\mathbf{x},\mathbf{y})=d_{01}(\mathbf{x},\mathbf{y})=k-\lambda$ by Lemma \ref{SBIBD-lem}.
Therefore, the frequency of $(d_{10}(\mathbf{x},\mathbf{y}),d_{01}(\mathbf{x},\mathbf{y}))=(k-\lambda,k-\lambda)$ is $2v(v-1)$.

When $\mathbf{x}\in {\cal C}_2^{(1)}$ and $\mathbf{y}\in {\cal C}_2^{(2)}$, there are two cases: either the support set of $\mathbf{y}$ is the complement of that of $\mathbf{x}$, or not.
The first case, if ${\rm supp}(\mathbf{y})=[v]\setminus {\rm supp}(\mathbf{x})$, then $(d_{10}(\mathbf{x},\mathbf{y}),d_{01}(\mathbf{x},\mathbf{y}))=(k,v-k)$, and the frequency of such pairs is $v$.
By symmetry, the frequency of $(d_{10}(\mathbf{y},\mathbf{x}),d_{01}(\mathbf{y},\mathbf{x}))=(v-k,k)$ is also $v$.
The second case, if ${\rm supp}(\mathbf{y})\ne [v]\setminus {\rm supp}(\mathbf{x})$, then $d_{11}(\mathbf{x},\mathbf{y})=k-\lambda$ by Lemma \ref{SBIBD-lem}.
Thus we have $(d_{10}(\mathbf{x},\mathbf{y}),d_{01}(\mathbf{x},\mathbf{y}))=(\lambda,v-2k+\lambda)$ and $(d_{10}(\mathbf{y},\mathbf{x}),d_{01}(\mathbf{y},\mathbf{x}))=(v-2k+\lambda,\lambda)$.
Both of them have the same frequency $v(v-1)$.

\noindent
\textbf{The code ${\cal C}_3$.} 
We decompose ${\cal C}_3$ into ${\cal C}_3^{(1)}$ and ${\cal C}_3^{(2)}$,
where the support set families of ${\cal C}_3^{(1)}$, ${\cal C}_3^{(2)}$ are $\big\{B\cup\{g_{v+1}\}: B\in{\cal B}\big\}$ and $\bar{\cal B}$, respectively.
It follows that $|{\cal C}_3^{(1)}|=|{\cal C}_3^{(2)}|=v$.
Similar to ${\cal C}_2$, if $\mathbf{x},\mathbf{y} \in {\cal C}_3^{(1)}$ (or ${\cal C}_3^{(2)}$), we have $d_{10}(\mathbf{x},\mathbf{y})=d_{01}(\mathbf{x},\mathbf{y})=k-\lambda$,
and the frequency of such pair is $2v(v-1)$.
When $\mathbf{x}\in {\cal C}_3^{(1)}$, $\mathbf{y}\in {\cal C}_3^{(2)}$ and ${\rm supp}(\mathbf{y})=[v+1]\setminus {\rm supp}(\mathbf{x})$,
then $(d_{10}(\mathbf{x},\mathbf{y}),d_{01}(\mathbf{x},\mathbf{y}))=(k+1,v-k)$ with frequency $v$.
Reversing the order of the two codewords gives $(v-k,k+1)$ also with frequency $v$.
And if ${\rm supp}(\mathbf{y})\ne [v+1]\setminus {\rm supp}(\mathbf{x})$, then $d_{11}(\mathbf{x},\mathbf{y})=k-\lambda$ since the additionally added point does not contribute to it.
Thus we have $(d_{10}(\mathbf{x},\mathbf{y}),d_{01}(\mathbf{x},\mathbf{y}))=(\lambda+1,v-2k+\lambda)$ and $(d_{10}(\mathbf{y},\mathbf{x}),d_{01}(\mathbf{y},\mathbf{x}))=(v-2k+\lambda,\lambda+1)$.
Each of these two bidistances has frequency $v(v-1)$.

\noindent
\textbf{The code ${\cal C}_4$.} 
Similar to ${\cal C}_2$ and ${\cal C}_3$, we decompose ${\cal C}_4$ into ${\cal C}_4^{(1)}$ and ${\cal C}_4^{(2)}$,
where the support set families of ${\cal C}_4^{(1)}$, ${\cal C}_4^{(2)}$ are $\big\{B\setminus\{g\}: B\in{\cal B},g\in B\big\}$ and $\big\{\bar{B}\setminus\{g\}: B\in{\cal B},g\notin B \big\}$, respectively.
Then we have $|{\cal C}_4^{(1)}|=k$ and $|{\cal C}_4^{(2)}|=v-k$ by Lemma \ref{SBIBD-lem}.
Similar to ${\cal C}_2$ and ${\cal C}_3$, if $\mathbf{x},\mathbf{y} \in {\cal C}_4^{(1)}$ (or ${\cal C}_4^{(2)}$), 
then $d_{10}(\mathbf{x},\mathbf{y})=d_{01}(\mathbf{x},\mathbf{y})=k-\lambda$, and the frequency of such pair could be derived from the number of blocks containing, 
respectively not containing, the fixed point $g$, i.e., $k(k-1)+(v-k)(v-k-1)$.
If $\mathbf{x}\in {\cal C}_4^{(1)}$ and $\mathbf{y}\in {\cal C}_4^{(2)}$, then $d_{11}(\mathbf{x},\mathbf{y})=k-\lambda -1$ since the fixed point $g$ has been removed.
Thus $d_{10}(\mathbf{x},\mathbf{y})=\lambda$ and $d_{01}(\mathbf{x},\mathbf{y})=v-2k+\lambda$, and the frequency of such pair is $k(v-k)$.
Reversing the order of the two codewords gives $(v-2k+\lambda,\lambda)$ also with frequency $k(v-k)$.
\end{proof}
\begin{remark}
Similar to Tables \ref{tab:2weight} and \ref{tab:3weight}, the AHB values and their frequencies listed in Table \ref{tab:SBIBD} are also given in a formal casewise manner.
In specific cases, some of these AHB entries may coincide and should be merged accordingly.
For example, for the code ${\cal C}_3$, when $v=2k+1$, we have $v-k=k+1$, $v-2k+\lambda=\lambda+1$ and $k-\lambda=\lambda+1$.
\end{remark}

\section{Conclusion}\label{sec:7}

In this paper, we introduced the concept of asymmetric Hamming bidistance (AHB) and its two-dimensional distribution as a refined characterization tool for binary codes operating over asymmetric channels. 
This notion addresses a key limitation of existing discrepancy-based measures \cite{GR2022},
which fail to distinguish codes that share identical weight distributions and minimum (symmetric) discrepancies but exhibit different decoding performance under maximum-likelihood decoding (MLD).

We first established the relationship between the AHB distribution and the average decoding error probability, 
and derived a new upper bound that is generally incomparable with the two known bounds from \cite{GR2022}. 
This bound offers enhanced discriminative power in scenarios where conventional measures prove insufficient.

To demonstrate the computability of the asymmetric bidistance in several classical code families, we computed their complete AHB distributions. 
Using strongly regular graphs, we determined these distributions for binary two-weight projective codes (excluding the zero codeword). 
By employing 3-class association schemes, we extended the analysis to specific binary three-weight projective codes under the same exclusion. 
Furthermore, utilizing the properties of symmetric balanced incomplete block designs (SBIBDs), we obtained the AHB distributions for several classes of binary nonlinear codes constructed in \cite{ZWC2025}.

These results not only validate the theoretical framework established in this work but also show that the AHB distribution provides a more refined analytical tool,
capable of yielding sharper error bounds in certain cases.
Future work may explore the application of AHB distributions to more general code families and investigate their potential in code design optimized for binary asymmetric channels.


\begin{thebibliography}{1}
\bibliographystyle{IEEEtran}

\bibitem{book1968} E. R. Berlekamp, Algebraic Coding Theory, New York: McGraw-Hill, pp. 397-399, 1968.
\bibitem{CDS1999} T. Beth, D. Jungnickel and H. Lenz, Design theory. Cambridge, U.K.: Cambridge University Press, 1999.
\bibitem{CK1986} R. Calderbank and W. M. Kantor, ``The geometry of two-weight codes",  \textit{Bull. Lond. Math. Soc.}, vol. 18, pp. 97–122, 1986.
\bibitem{CGO1999} P. Cappelletti, C. Golla, P. Olivo, and E. Zanoni, ``Flash Memories," Boston, MA, USA: Kluwer, 1999.
\bibitem{CSB2010} Y. Cassuto, M. Schwartz, V. Bohossian, and J. Bruck, ``Codes for asymmetric limited-magnitude errors with application to multilevel flash memories,"
\textit{IEEE Trans. Inf. Theory}, vol. 56, no. 4, pp. 1582–1595, 2010.
\bibitem{CR1979} S. D. Constantin and T. R. N. Rao, ``On the theory of binary asymmetric error correcting codes," \textit{Inf. Control}, vol. 40, no. 1, pp. 20-36, 1979.
\bibitem{GR2022} G. Cotardo and A. Ravagnani, ``Parameters of codes for the binary asymmetric channel," \textit{IEEE Trans. Inf. Theory,} vol. 68, no. 5, pp. 2941-2950, 2022.
\bibitem{CIM2013} C. Curto, V. Itskov, K. Morrison, Z. Roth, and J. L. Walker, ``Combinatorial neural codes from a mathematical coding theory perspective,"
\textit{Neural Comput.}, vol. 25, no. 7, pp. 1891–1925, 2013.
\bibitem{Delsarte1973} P. Delsarte, ``An algebraic approach to the association schemes of coding theory," \textit{Philips Res. Repts.}, 1973.
\bibitem{Delsarte1972} P. Delsarte, ``Weights of linear codes and strongly regular normed spaces," \textit{Discrete Mathematics}, vol. 3, pp. 47-64, 1972.
\bibitem{DCC2006} A. Faldum, J. Lafuente, G. Ochoa and W. Willems, ``Error probabilities for bounded distance decoding," \textit{Des., Codes Crypt.,} vol. 40, pp. 237–252, 2006. 
\bibitem{GD2012} R. Gabrys and L. Dolecek, ``Coding for the binary asymmetric channel," 2012 International Conference on Computing, Networking and Communications (ICNC), Maui, HI, USA, pp. 461-465, 2012.
\bibitem{TK1981} T. Kløve, ``Error correcting codes for the asymmetric channel," Dept. Inform., Univ. Bergen, Bergen, Norway, Tech. Rep., 1981.
\bibitem{optical2007} H.S. Mruthyunjaya, ``Performance improvement of all optical WDM systems on binary asymmetric channel," \textit{Int. J. Microw. Opt. Technol.}, vol. 2, no. 3, pp. 230-235, 2007.
\bibitem{MGK1998} I. S. Moskowitz, S. J. Greenwald, and M. H. Kang, ``An analysis of the timed Z-channel," \textit{IEEE Trans. Inf. Theory,} vol. 44, no. 7, pp. 3162–3168, 1998.
\bibitem{rep2010} S. M. Moser, P.-N. Chen, H.-Y. Lin, ``Error probability analysis of binary asymmetric channels," National Chiao Tung University, Tech. Rep., 2010.
\bibitem{P1994} G. Poltyrev, ``Bounds on the decoding error probability of binary linear codes via their spectra," \textit{IEEE Trans. Inf. Theory,} vol. 40, no. 4, pp. 1284-1292, 1994.
\bibitem{P2016} A. Poplawski, ``On matched metric and channel problem," 2016, arXiv:1606.02763.
\bibitem{QCR2018} C. Qureshi, S. I. R. Costa, C. B. Rodrigues and M. Firer, ``On equivalence of binary asymmetric channels regarding the maximum likelihood decoding," 
\textit{IEEE Trans. Inf. Theory}, vol. 64, no. 5, pp. 3528-3537, 2018.
\bibitem{Q2019} C. M. Qureshi, ``Matched metrics to the binary asymmetric channels," \textit{IEEE Trans. Inf. Theory,} vol. 65, no. 2, pp. 1106–1112, 2019.
\bibitem{S1955} R. Silverman, ``On binary channels and their cascades," \textit{IRE Trans. Inf. Theory,} vol. 1, no. 3, pp. 19-27, 1955.
\bibitem{SSL2026} Z. Sun, X. Song and Y. Li, ``Constructions of combinatorial neural codes with asymmetric discrepancy," \textit{IEEE Trans. Inf. Theory,}  vol. 72, no. 5, pp. 2949-2962, 2026.
\bibitem{XL2005} C. Xing and J. Ling, ``A construction of binary constant-weight codes from algebraic curves over finite fields," \textit{IEEE Trans. Inf. Theory,} vol. 51, no. 10,  pp. 3674-3678, 2005.
\bibitem{ZJF2023} A. Zhang, X. Jing and K. Feng, ``Optimal combinatorial neural codes with matched metric $\delta_r$: characterization and constructions," 
\textit{IEEE Trans. Inf. Theory,} vol. 69, pp. 5440–5448, 2023.
\bibitem{ZWC2025} X. Zheng, S. Wang and C. Fan, ``Optimal combinatorial neural codes via symmetric designs," \textit{Des. Codes Cryptogr.}, vol. 93, pp. 725–736, 2025. 
\end{thebibliography}
\end{document}